\documentclass[12pt,reqno]{amsart}

\usepackage{amsthm,amsmath,amsfonts,amssymb}
\usepackage{mathrsfs}
\usepackage{ascmac}
\usepackage{bm}
\usepackage[numbers]{natbib}
\usepackage[dvipdfmx]{graphicx}
\usepackage{graphicx,here,comment,subcaption}
\usepackage{amsaddr}
\usepackage{fancyhdr}
\usepackage{url}
\usepackage{booktabs}

\usepackage{fullpage}
\newtheorem{theorem}{Theorem}

\newtheorem{lemma}{Lemma}
\newtheorem{proposition}{Proposition}

\theoremstyle{definition}
\newtheorem{definition}{Definition}

\renewcommand{\hat}{\widehat}
\renewcommand{\tilde}{\widetilde}

\usepackage{color}

\usepackage{amsaddr}
\title{Inference for Projection-Based Wasserstein Distances on Finite spaces
}

\author{Ryo Okano$^1$ \and Masaaki Imaizumi$^{1,2}$}

\allowdisplaybreaks

\begin{document}
\maketitle

\begin{center}
    $^1$The  University of Tokyo / $^{2}$RIKEN Center for Advanced Intelligence Project
\end{center}

\begin{abstract}
   The Wasserstein distance is a distance between two probability distributions and has recently gained increasing popularity in statistics and machine learning, owing to its attractive properties. One important approach to extending this distance is using low-dimensional projections of distributions to avoid a high computational cost and the curse of dimensionality in empirical estimation, such as the sliced Wasserstein or max-sliced Wasserstein distances.
   Despite their practical success in machine learning tasks, the availability of statistical inferences for projection-based Wasserstein distances is limited owing to the lack of distributional limit results.
   In this paper, we consider distances defined by integrating or maximizing Wasserstein distances between low-dimensional projections of two probability distributions. Then we derive limit distributions regarding these distances when the two distributions are supported on finite points. We also propose a bootstrap procedure to estimate quantiles of limit distributions from data. This facilitates asymptotically exact interval estimation and hypothesis testing for these distances. 
   Our theoretical results are based on the arguments of Sommerfeld and Munk (2018) for deriving distributional limits regarding the original Wasserstein distance on finite spaces and the theory of sensitivity analysis in nonlinear programming.
   Finally, we conduct numerical experiments to illustrate the theoretical results and demonstrate the applicability of our inferential methods to real data analysis.
\end{abstract}

\section{Introduction}

The Wasserstein distance is a distance between two probability distributions, having attracted considerable interest in the statistics and machine learning literature \cite{villani2009optimal, panaretos2019statistical, peyre2019computational}. This distance is based on the optimal transport problem and measures the amount of work required to transform one distribution into another. 
Specifically, given two probability distributions $P$ and $Q$ with finite $p \ge 1$ moments and support in $\mathcal{X} \subset \mathbb{R}^d$, $d \ge 1$, the $p$-Wasserstein distance between $P$ and $Q$ is defined as
\begin{equation}
    {W}_p(P, Q) = \left(\inf_{\pi \in \Pi(P, Q)} \int_{\mathcal{X} \times \mathcal{X}} \|x - y\|^p d \pi(x, y) \right)^{1/p},
\end{equation}
where $\Pi(P, Q)$ is the set of joint probability distributions whose respective marginals coincide with $P$ and $Q$, known as couplings.
Compared to other measures of distribution closeness, such as the Kullback--Leibler divergence or the total variation distance, the Wasserstein distance has two main advantages: (i) it is sensitive to the underlying geometry of distribution support and (ii) it does not assume the absolute continuity of distributions with respect to the other. Owing to these advantages, it has recently been used as an attractive data analytical tool, particularly in computer vision \cite{rubner2000earth, solomon2015convolutional, sandler2011nonnegative}
and natural language processing \cite{kusner2015word, zhang2016building}.

Recently, various extensions of the original Wasserstein distance have been proposed to address its shortcomings, mainly the high computational cost and the curse of dimensionality in empirical estimation \citep{peyre2019computational, weed2019sharp}. One important approach is using low-dimensional projections of distributions, that is, computing the Wasserstein distances between low-dimensional projections of distributions $P$ and $Q$ instead of dealing with the original ones. 
The most representative example of this approach is the sliced Wasserstein distance \cite{rabin2011wasserstein, bonneel2015sliced}, which averages the Wasserstein distances between the random one-dimensional projections of distributions $P$ and $Q$.
As the Wasserstein distance between univariate distributions is known to be easily computed, the sliced Wasserstein distance is an easily computable variant.
Another example is the max-sliced Wasserstein distance \cite{deshpande2019max}, which maximizes the Wasserstein distances between random one-dimensional projections and has a computational advantage. By considering $k$-dimensional projections $(1 \le k \le d)$,
the max-sliced Wasserstein distance is generalized to the 
 projection robust Wasserstein (PRW) distance \cite{paty2019subspace,niles2019estimation}.
The PRW distance effectively captures the difference between two distributions if they only differ in a low-dimensional subspace and solves the curse of dimensionality in estimation \citep{niles2019estimation, lin2021projection}.
Several recent studies have shown that these proposals are practical for many machine learning tasks \cite{lin2020projection, kolouri2016sliced, kolouri2018sliced, carriere2017sliced, liutkus2019sliced}.

The development of inferential tools (e.g., interval estimation or hypothesis testing) for the Wasserstein distance and its extensions has become an active research area in statistics. Many studies have derived limit distributions for these distances as a basis for inferential procedures.
For example, the limit distributions of the empirical Wasserstein distance are studied when distributions $P$ and $Q$ are supported in $\mathbb{R}$ \cite{munk1998nonparametric, freitag2005hadamard, del1999tests, ramdas2017wasserstein} and when they are supported at finite or countable points \cite{sommerfeld2018inference, tameling2019empirical}. 
The limit distributions of the empirical regularized optimal transport distance in finite spaces, which is an easily computable extension of the Wasserstein distance, have been derived by
\cite{bigot2019central, klatt2020empirical}. 
However, for projection-based extensions of the Wasserstein distance, such distributional limit results are not well established, which hinders their inference.
For more details on related works, see Section \ref{sec:related}.

In this study, we propose inferential procedures for projection-based Wasserstein distances when distributions $P$ and $Q$ are supported on finite points.
We consider two types of distances: (i) the integral projection robust Wasserstein (IPRW) distance, which is defined by integrating Wasserstein distances between $k$-dimensional projections of distributions $P$ and $Q$ $(1 \le k \le d)$ and includes the sliced Wasserstein distance as a special case, and (ii) the PRW distance we introduced above.
As a first contribution, we derive limit distributions of the empirical IPRW and PRW distances with entropic regularization.
Second, we show the consistency of the rescaled bootstrap (or the $m$-out-$n$ bootstrap), which enables us to estimate quantiles of the limit distributions from data. 
Consequently, we construct asymptotically exact confidence intervals for these two distances, and obtain new statistics for testing the equality of the distributions. Finally, we conduct numerical experiments to verify our theoretical results and apply our inferential methods to real data analysis.

We derive our distributional limits by showing the directional Hadamard differentiability of the IPRW and PRW distances and applying a refined delta method. This strategy was developed by \cite{sommerfeld2018inference} for the inference of the original Wasserstein distance in finite spaces.  
To implement this strategy for the PRW distance, we use the following two key techniques. 
First, we utilize sensitivity analysis in nonlinear programming, which investigates how the optimal value of an optimization problem changes when the objective function and the constraints are changed \cite{fiacco1983introduction}.
We regard the PRW distance between distributions $P$ and $Q$ as the optimal value of a parametric optimization problem with parameters $P$ and $Q$, and apply the result of the sensitivity analysis to show its directional differentiability.
Second, we introduce an entropic regularization term in the PRW distance to avoid including a non-smooth objective function in its definition, which helps show its directional differentiability.
The idea of adding an entropic regularization term to the PRW distance was proposed by \cite{lin2020projection}, and we call this quantity the regularized PRW distance.

We can summarize the contributions of this paper as follows:
\begin{itemize}
    \item We derive limit distributions of the empirical versions of the IPRW and regularized PRW distances when distributions $P$ and $Q$ are supported on finite points. 
    \item We show the consistency of the rescaled bootstrap for the IPRW and PRW distances, 
    which enables us to estimate quantiles of the limit distributions from  data.
    This facilitates asymptotically exact interval estimations and hypothesis testing for these distances. 
    \item We conduct numerical experiments to illustrate our theoretical  results, and show the applicability of our inferential methods to real data analysis.
\end{itemize}

\subsection{Related work} \label{sec:related}
 There are several extensions of the Wasserstein distance based on low-dimensional projections, in addition to the distances we consider, such as the generalized sliced \cite{kolouri2019generalized}, tree-sliced \cite{le2019tree}, and distributional sliced \cite{nguyen2020distributional} Wasserstein distances. 
Beyond the projection-based approach, 
\cite{cuturi2013sinkhorn} proposed the entropic regularization of optimal transport, which can be efficiently computed through an iterative method, called the Sinkhorn algorithm. Further, \cite{goldfeld2020gaussian} proposed the smooth Wasserstein distance, which avoids the curse of dimensionality in estimation by smoothing out local irregularities in distributions $P$ and $Q$ via convolution with a Gaussian kernel.

Statistical inference for the Wasserstein distance and its extensions has been studied in several settings, based on their limit distributions. When distributions $P$ and $Q$ are supported in $\mathbb{R}$, the Wasserstein distance between them has a closed form and is described as the $L^p$ norm of the quantile functions of $P$ and $Q$. Using this fact, \cite{munk1998nonparametric, freitag2005hadamard, del2005asymptotics, ramdas2017wasserstein} derived the limit distributions of the empirical Wasserstein distances in the univariate case and studied the validity of the bootstrap. The inference for the Wasserstein distance over finite spaces was studied by \cite{sommerfeld2018inference}, and the result was extended to a case with countable spaces by \cite{tameling2019empirical}.
\cite{bigot2019central, klatt2020empirical} considered inference for the entropic regularized optimal transport distance on finite spaces.
In a general setting, \cite{del2019central} established central limit theorems for the empirical Wasserstein distance and \cite{mena2019statistical} established similar results for the entropic regularized optimal transport distance; however, these results contain unknown centering constants that hinder their use for statistical inference.  

To the best of our knowledge, statistical inference for projection-based Wasserstein distances has only been considered in one study. Specifically,
\cite{manole2019minimax} proposed confidence intervals with finite-sample validity for the sliced Wasserstein distance and showed their minimax optimality in length.
Owing to the closed-form expression of the one-dimensional Wasserstein distance, their inference method is valid, without imposing strong assumptions on distributions $P$ and $Q$ such as discreteness.
However, their approach is not applicable when the projection dimension is greater than one. 
 By contrast, our work covers Wasserstein distances based on projections to dimensions greater than one. 

\subsection{Notation}
$\|\cdot \|$ and $\langle \cdot \rangle$ denote the Euclidean norm and inner product, respectively. $\mathbb{R}_{> 0}$ is the positive real and $\mathbb{R}_{\ge 0}$ the non-negative real. $\otimes$ is the Kronecker product.
For any $a, b \in \mathbb{R}$, $a \land b$ denotes the minima of $a$ and $b$. 
For $1 \le k \le d$, the set of $d \times k$ matrices with orthonormal columns is denoted as $S_{d, k} = \{E \in \mathbb{R}^{d \times k}: E^\top E = I_k\}$. Note that, when $k=1$, $S_{d, k}$ coincides with the $d$-dimensional unit ball, $\mathbb{S}^{d-1} = \{ x \in \mathbb{R}^d: \|x\| = 1 \}$. 
Given a map $T: \mathbb{R}^d \to \mathbb{R}$ and Borel probability measure $P$ supported in $\mathbb{R}^d$, $T_{\#}P$ denotes the pushforward of $P$ under $T$, defined by $T_{\#}P(B) = P(T^{-1}(B))$ for all Borel sets $B \subset \mathbb{R}^d$.
For any set $A \subset \mathbb{R}^d$, its diameter is denoted by $\text{diam}(A) = \sup \{\|x-y\|: x, y \in A \}$.
$\mathcal{P}(\mathbb{R}^n)$ denotes the set of all subsets of $\mathbb{R}^n$.
$\stackrel{d}{\to}$ denotes convergence in  distribution of random variables and $\stackrel{d}{=}$ denotes distributional equality of the random variables.

\section{Background}
Here, we provide background details on the Wasserstein distance and its projection-based extensions in finite spaces.

\subsection{Wasserstein distance and entropic regularization}
\subsubsection{Wasserstein distance}
In this study, we restrict support $\mathcal{X} = \{x_1, ..., x_N\} \subset \mathbb{R}^d$ to a finite set of size $N \in \mathbb{N}$.
Every probability measure on $\mathcal{X}$ is represented as an element in an $(N-1)$-dimensional sphere $\Delta_N = \{ r \in \mathbb{R}^N_{> 0} : \sum_{i=1}^N r_i = 1 \}$; hence, we do not distinguish vector $r\in \Delta_N$ and its corresponding probability distribution.
Given support $\mathcal{X} = \{x_1, ..., x_N\}$ and order $p \ge 1$, we define cost vector $c_p(\mathcal{X}) \in \mathbb{R}^{N^2}$ as
${c_p(\mathcal{X})}_{(i-1)N+j} = \|x_i - x_j \|^p$ for $1 \le i, j \le N$, representing the transport cost from $x_i$ to $x_j$.
The $p$-Wasserstein distance between the two distributions $r, s \in \Delta_N$ on $\mathcal{X} \subset \mathbb{R}^d$ is given by
\begin{equation}
    W_p(r, s; \mathcal{X}) = \left\{ \min_{\pi \in \Pi(r, s)} \langle c_p(\mathcal{X}), \pi \rangle \right\}^{1/p},
    \label{def_wasserstein}
\end{equation}
where $\Pi(r, s)$ is a set of vectors of length $N^2$ that represent the couplings of $r$ and $s$.
Formally, $\Pi(r, s)$ is defined as
\begin{equation}
\label{transport}
\Pi(r, s) = \left\{\pi \in \mathbb{R}^{N^2}: A\pi = 
\left(
\begin{array}{c}
r \\
s 
\end{array}
\right)
 \right\},
\end{equation}
where $A$ is a coefficient matrix:
\begin{equation*}
    A = \left(
\begin{array}{c}
I_{N \times N} \otimes 1_{1 \times N} \\
1_{1 \times N} \otimes I_{N \times N} 
\end{array}
\right)
\in \mathbb{R}^{2N \times N^2}.
\end{equation*}
Constraint $A\pi = (r, s)^\top$ ensures that $\pi$ satisfies the marginal constraints: a matrix $\tilde{\pi} \in \mathbb{R}^{N \times N}$, generated by $\pi$ as $ \tilde{\pi}_{i,j} = \pi_{(i-1)N + j}$, satisfies 
$\sum_{j=1}^N\tilde{\pi}_{i, j} = r_i$ for $1 \le i \le N$, and $\sum_{i=1}^N\tilde{\pi}_{i, j} = s_j$  for $1 \le j \le N$.

\subsubsection{Entropic regularization}
The entropic regularization is a typical extension of the Wasserstein distance \cite{cuturi2013sinkhorn}.
Given $p \ge 1$, distributions $r, s \in \Delta_N$, and a regularization parameter $\lambda > 0$, we consider an entropic regularized optimal transport problem as follows:
\begin{equation}
\label{EOT}
    \min_{\pi \in \Pi(r, s)} \langle c_p(\mathcal{X}), \pi \rangle + \lambda \varphi(\pi),
\end{equation}
where $\varphi: \mathbb{R}^{N^2} \to \mathbb{R}$ is the negative Boltzmann-Shannon entropy, defined as
\begin{equation}
    \varphi(\pi) = 
\begin{cases}
\sum_{i=1}^{N^2}\pi_i \log(\pi_i) - \pi_i + 1 & \text{if} \,\ \pi \in \mathbb{R}^{N^2}_{\ge0},\\
 + \infty & \text{otherwise}. 
\end{cases}
\end{equation}
Here, we set $0 \log(0) = 0$. 
Because problem \eqref{EOT} is a strictly convex optimization problem, it has a unique optimal solution. 
We refer to the solution of \eqref{EOT} as the regularized optimal transport plan $\pi_{p, \lambda}(r, s ;\mathcal{X})$. 
Using this notion, we can define the $p$-regularized optimal transport distance (or the $p$-Sinkhorn divergence)
between $r, s \in \Delta_N$ as
\begin{equation}
    W_{p, \lambda}(r, s;\mathcal{X}) = \langle c_p(\mathcal{X}), \pi_{p, \lambda}(r, s; \mathcal{X}) \rangle^{1/p}.
\end{equation}
Several computational advantages and statistical properties of the regularized optimal transport distance have been studied (e.g., see \cite{cuturi2013sinkhorn,peyre2019computational, klatt2020empirical, bigot2019central}).

\subsection{Projection-based Wasserstein distances}
We introduce extensions of the Wasserstein distance based on low-dimensional projections of the distributions. 
Fix $k \leq d$ and let $\pi_E: x \in \mathbb{R}^d \mapsto E^\top x$ for $E \in S_{d, k}$. 
For distribution $P$ on $\mathbb{R}^d$, the $k$-dimensional projection of $P$ in $E \in S_{d, k}$ is defined by $P_E = {\pi_E}_{\#} P$. That is, $P_E$ is the distribution of $E^\top X$ for $X \sim P$. 

\subsubsection{Integral projection robust Wasserstein distance}
We study $k$-dimensional projections of distributions $r, s \in \Delta_N$ on a finite $\mathcal{X} = \{x_1, ..., x_N \} \subset \mathbb{R}^d$. 
The Wasserstein distance between the projections of the distributions $r$ and $s$ in direction $E \in S_{d, k}$ is
represented by $W_p(r, s; \mathcal{X}_E)$, where $\mathcal{X}_E = \{E^\top x_1, ..., E^\top x_N\} \subset \mathbb{R}^k$. 
The $p$-integral projection robust Wasserstein (IPRW) distance \cite{lin2021projection} is defined as an integral of the Wasserstein distances over the direction $E$, that is,
\begin{equation}
    \mathrm{IW}_p(r, s) = \left(\int_{S_{d,k}} W_p^p(r, s; \mathcal{X}_E) d \mu(E) \right)^{1/p},
\end{equation}
where $\mu$ is a given measure on $S_{d, k}$.
\cite{lin2021projection} shows that the IPRW distance with the uniform measure on $S_{d, k}$ solves the curse of dimensionality in estimation.
When the projection dimension is $k = 1$ and $\mu$ is a uniform measure of $S_{d, 1}$, which coincides with the $d$-dimensional unit ball $\mathbb{S}^{d-1}$, the IPRW distance corresponds to the sliced Wasserstein distance \cite{rabin2011wasserstein, bonneel2015sliced}.
The sliced Wasserstein distance has the advantage of being easy to calculate, due to the fact that the Wasserstein distance between one-dimensional distributions is easy to compute.

\subsubsection{Projection robust Wasserstein distance}
The $p$-projection robust Wasserstein (PRW) distance \cite{paty2019subspace} is defined as the maximum Wasserstein distance between $k$-dimensional projections of $r, s \in \Delta_N$ over direction $E \in S_{d,k}$, namely
\begin{equation}
\mathrm{PW}_p(r, s) = \max_{E \in S_{d, k}} W_p(r, s; \mathcal{X}_E). 
\end{equation}
When $k = 1$, the PRW distance corresponds to the max-sliced Wasserstein distance \cite{deshpande2019max}.
The PRW distance
effectively captures the difference between the two distributions $r, s$ if they differ only in a low-dimensional subspace, and
 \cite{niles2019estimation, lin2021projection} showed that it solves the curse of dimensionality in estimation.

We further introduce the entropic regularization for the PRW distance.
With a fixed regularization parameter $\lambda > 0$ and projection direction $E \in S_{d, k}$, we represent the regularized optimal transport distance between the projections of $r$ and $s$ as $W_{p, \lambda}(r, s; \mathcal{X}_E)$.
Then, the $p$-regularized PRW distance is defined by
\begin{equation}
\mathrm{PW}_{p, \lambda}(r, s) = \max_{E \in S_{d, k}} W_{p, \lambda}(r, s; \mathcal{X}_E).
\end{equation}
This method with entropy regularization has the advantage of reducing the computational cost, owing to smoothing out the non-smoothness due to maximization \cite{lin2020projection}.

\section{Distributional limits}
We study distributional limits of the empirical version of the IPRW and regularized PRW distances on a finite space. 
Specifically, we consider the following setting.
For probability distributions $r, s \in \Delta_N$ on $\mathcal{X} = \{x_1, ..., x_N\} \subset \mathbb{R}^d$ and sample sizes $n$ and $m$, let $X_1, ..., X_n \sim r, Y_1, ..., Y_m \sim s$ be independent and identically distributed (i.i.d.) samples. Then, we define their corresponding empirical distributions $\hat{r}_n, \hat{s}_m \in \Delta_N$, whose $i$th elements are given as
\[
\hat{r}_{n, i} = \frac{\#\{k: X_k = x_i\}}{n}, \quad
\hat{s}_{m, i} = \frac{\#\{k: Y_k = x_i\}}{m},
\]
for $1 \le i \le N$.
Given order $p \ge 1$ and regularization parameter $\lambda > 0$, we derive the distributions to which
\[
\sqrt{\frac{nm}{n+m}} \{ \mathrm{IW}_p(\hat{r}_n, \hat{s}_m) - \mathrm{IW}_p(r, s) \}
\]
and
\[
\sqrt{\frac{nm}{n+m}} \{ \mathrm{PW}_{p, \lambda}(\hat{r}_n, \hat{s}_m) - \mathrm{PW}_{p, \lambda}(r, s) \}
\]
converge in law as $n, m \to \infty$.
\subsection{Outline and preparation}
We derive distributional limits using the delta method, which is based on the differentiability of the IPRW and regularized PRW distances. 
Specifically, following that
$
\sqrt{\frac{nm}{n + m}}
\{
(\hat{r}_n, \hat{s}_m)
-
(r, s)
\}
$
converges to a Gaussian random vector by the central limit theorem, we can derive  distributional limits by applying the delta method
with the maps
  $(r, s) \mapsto \mathrm{IW}_p(r, s)$ and $(r, s) \mapsto\mathrm{PW}_{p, \lambda}(r, s)$.
To use the delta method in this setting, we consider the directional Hadamard differentiability, which is defined as follows.
\begin{definition}[Directional Hadamard differentiability \cite{romisch2004delta, sommerfeld2018inference}]
A function $f:D_f \subset \mathbb{R}^d \to \mathbb{R}$ is directionally Hadamard differentiable at $u \in D_f$ tangentially to $D_0 \subset \mathbb{R}^d$, if there exists a map $f_u': D_0 \to \mathbb{R}$ so that
\begin{equation}
\label{limit}
    \lim_{n \to \infty} \frac{f(u + t_n h_n) - f(u)}{t_n}
=
f_u'(h),
\end{equation}
for any $h \in D_0$ and arbitrary sequences $\{t_n\}\subset \mathbb{R}$ and $\{h_n\} \subset \mathbb{R}^d$ so that $t_n \searrow 0$, $h_n \to h$, and $u + t_n h_n \in D_f$ for all large $n \in \mathbb{N}$.
We refer $f'_u$ to the directional Hadamard derivative.
\end{definition}

In contrast to the usual (non-directional) Hadamard differentiability (e.g., \cite{van2000asymptotic}), directional Hadamard differentiability does not require the derivative to be linear, but allows for the
 Delta method.
\begin{theorem}[Delta method with a directionally Hadamard differentiable map: Theorem  1 in \cite{romisch2004delta} and Theorem 3 in \cite{sommerfeld2018inference}]
Let $f: D_f \subset \mathbb{R}^d \to \mathbb{R}$ be directionally Hadamard  differentiable at $u \in D_f$ tangentially to $D_0 \subset \mathbb{R}^d$ with derivative $f_u': D_0 \to \mathbb{R}$. Let $T_n$ be $\mathbb{R}^d$-valued random variables, so that $\rho_n(T_n - u) \stackrel{d}{\to} T$ for a sequence of numbers $\rho_n \to \infty$, and a random variable $T$ taking its values in $D_0$. Then, $\rho_n(f(T_n) - f(u)) \stackrel{d}{\to} f_u'(T)$.
\label{delta_method}
\end{theorem}

Our approach using the directional Hadamard derivative is important when dealing with the projection-based Wasserstein distances. 
These distances are not differentiable in the sense of (non-directional) Hadamard differentiation, but will be shown to have a directional Hadamard derivative, which makes it possible to apply the delta method.

\subsection{Distributional limit for IPRW distance}
As our first main result, we derive a distributional limit of the empirical IPRW distance, $\mathrm{IW}_p(\hat{r}_n, \hat{s}_m)$.
To this end, we first show the directional Hadamard differentiability of the map $(r, s) \mapsto \mathrm{IW}_p^p(r, s)$ and derive its derivative. 
In preparation, we define sets of dual solutions for the optimization problem in \eqref{def_wasserstein}. 
Following \cite{sommerfeld2018inference}, given two distributions $r, s \in \Delta_N$ and a ground space $\mathcal{X} = \{x_1, ..., x_N\}$, we define  
\begin{equation}
\label{set1}
     \Phi_p^{\ast}(\mathcal{X})
     =
     \{u \in \mathbb{R}^N: u_i - u_j \le \|x_i - x_j\|^p, 1 \le i, j \le N\}
\end{equation}
and
\begin{equation}
\label{set2}
\begin{split}
    \Phi_p^{\ast}(r, s; \mathcal{X}) =\{(u, v) \in \mathbb{R}^{N} \times \mathbb{R}^N :
    \langle &u, r \rangle + \langle v, s \rangle = W_p^p(r, s; \mathcal{X}), \\
    &u_i + v_j \le \|x_i - x_j\|^p, 1 \le i, j \le N\}.
\end{split}
\end{equation}
These sets play a role in describing a limit distribution.
Additionally, we define a set of directions in which the limits are considered as $\Omega_N = \{h \in \mathbb{R}^N: \sum_{i=1}^N h_i = 0\}$.
Then, we achieve the following result on differentiability.
\begin{proposition}[Directional Hadamard differentiability of $\mathrm{IW}_p^p$]
 The map $\mathrm{IW}_p^p: \Delta_N \times \Delta_N \to \mathbb{R}, (r, s) \mapsto \mathrm{IW}_p^p(r, s)$ is directional Hadamard differentiable at all $(r, s) \in \Delta_N \times \Delta_N$ tangentially to $\Omega_N \times \Omega_N$ with derivative:
\begin{equation}
(h_1, h_2) \mapsto
     \int_{S_{d, k}}
\max_{(u, v) \in \Phi_p^{\ast}(r, s; \mathcal{X}_E)} 
\langle u, h_1 \rangle + \langle v, h_2 \rangle  
d\mu(E).
\end{equation}
\label{derive_IPRW}
\end{proposition}

\begin{proof}
Let $\{h_{1\ell}\}, \{h_{2\ell} \} \subset \Omega_N$ be sequences satisfying  $h_{1\ell} \to h_1, h_{2\ell} \to h_2$ and let $t_\ell \searrow 0$ as $\ell \to \infty$.
Following the definition of the directional Hadamard derivative, we consider the following difference:
\begin{equation}
\label{intgral}
\begin{split}
    &\frac{\mathrm{IW}_p^p(r + t_\ell h_{1\ell}, s + t_\ell h_{2\ell}) - \mathrm{IW}_p^p(r, s)}{t_\ell}. \\
&=
\int_{S_{d, k}} \frac{W_p^p(r + t_\ell h_{1\ell}, s + t_\ell h_{2\ell}; \mathcal{X}_E) - W_p^p(r, s; \mathcal{X}_E)}{t_\ell}
d \mu(E),
\end{split}
\end{equation}
and consider its limit.
For each $E \in S_{d, k}$, Theorem 4 in \cite{sommerfeld2018inference} implies \[
\frac{W_p^p(r + t_\ell h_{1\ell}, s + t_\ell h_{2\ell}; \mathcal{X}_E) - W_p^p(r, s; \mathcal{X}_E)}{t_\ell}
\to
\max_{(u, v) \in \Phi_p^{\ast}(r, s; \mathcal{X}_E)} 
\langle u, h_1 \rangle + \langle v, h_2 \rangle,
\]
as $\ell \to \infty$. Furthermore, the Lipschitz continuity of the Wasserstein distance (Theorem 4 of \cite{sommerfeld2018inference}) implies
\begin{align*}
    \left|\frac{W_p^p(r + t_\ell h_{1\ell}, s + t_\ell h_{2\ell}; \mathcal{X}_E) - W_p^p(r, s; \mathcal{X}_E)}{t_\ell}\right|
    &\le
    \frac{p \text{diam}(\mathcal{X}_E)^p\|t_\ell(h_{1\ell}, h_{2\ell})\|}{t_\ell}  \\
    & \le
    p k^p \text{diam}(\mathcal{X})^p \|(h_{1\ell}, h_{2\ell})\|.
\end{align*}
 The last inequality follows $\text{diam}(\mathcal{X}_E) \leq k\text{diam}(\mathcal{X})$, which follows
 \[
 \|E^\top x \|
 \le
 |E_1^\top x| + \cdots + |E_k^\top x|
 \le
 \|E_1^\top\| \|x\| + \cdots +  \|E_k^\top\| \|x\|
 =
 k\|x\|
 \]
 for $E = (E_1, ..., E_k) \in S_{d, k}$ and $x \in \mathbb{R}^d$.
Because $\mathcal{X}$ is finite and $\{h_{1\ell}\}$ and $\{h_{2\ell} \}$ are convergent sequences,
$pk^p \text{diam}(\mathcal{X})^p \|(h_{1\ell}, h_{2\ell})\|$ is bounded by a constant not depending on $E$ and $\ell$. Therefore, by taking $\ell \to \infty$ in (\ref{intgral}), we can apply the dominated convergence theorem, and the claim then holds.
\end{proof}

We state  our main result on a limit distribution of the empirical IPRW distance.
This derivation is based on the differentiability in Proposition \ref{derive_IPRW} and the delta method in Theorem \ref{delta_method}.
For $r \in \Delta_N$, we define the covariance matrix as
 \begin{equation}
    \Sigma(r) = 
\begin{pmatrix}
r_1(1-r_1)  & -r_1r_2     & \cdots  & -r_1r_{N}\\
-r_2r_1     & r_2(1-r_2)  & \cdots  & -r_2r_{N} \\
\vdots      & \vdots      &  \ddots &  \vdots \\
-r_{N}r_1 & -r_{N}r_2 & \cdots  & r_{N}(1-r_{N})
\end{pmatrix}. 
 \end{equation}
Then, we obtain the following result.
\begin{theorem}[Distributional limits of $\mathrm{IW}_p(\hat{r}_n, \hat{s}_m)$]
\label{asymdist1}
Let $r, s \in \Delta_{N}$ be two probability distributions supported on $\mathcal{X} \subset \mathbb{R}^d$, let $X_1, ..., X_n \sim r, Y_1, ..., Y_m \sim s$ be i.i.d. $n$ and $m$ samples, and let $\hat{r}_n , \hat{s}_m$ be the corresponding empirical distributions. 
Let $G \sim N(0, \Sigma(r))$ and $H \sim N(0, \Sigma(s))$ be independent Gaussian random vectors.
Then, we have the followings:
\begin{enumerate}
    \item If $r = s$, and $n \land m \to \infty$ and $m/(n+m) \to \delta \in (0, 1)$, we have
\[
\left(\frac{nm}{n+m} \right)^{\frac{1}{2p}} \mathrm{IW}_p(\hat{r}_n, \hat{s}_m) \stackrel{d}{\to} \left(\int_{S_{d, k}} \max_{u \in \Phi_p^{\ast}(\mathcal{X}_E)} \langle G, u \rangle d \mu(E)\right)^{1/p},
\]
where $\Phi_p^{\ast}(\mathcal{X}_E)$ is given by (\ref{set1}).
\item If $r \neq s$, and $n \land m \to \infty$ and $m/(n+m) \to \delta \in (0, 1)$, we have
\begin{equation*}
\begin{split}
    &\sqrt{\frac{nm}{n+m}} \{ \mathrm{IW}_p(\hat{r}_n, \hat{s}_m) - \mathrm{IW}_p(r, s) \} \\
\stackrel{d}{\to} &
\frac{1}{p}\mathrm{IW}_p^{1-p}(r, s) 
 \int_{S_{d, k}}
\max_{(u, v) \in \Phi_p^{\ast}(r, s; \mathcal{X}_E)} 
\sqrt{\delta}
\langle G, u \rangle + \sqrt{1-\delta} \langle H, v \rangle  
d\mu(E),
\end{split}
\end{equation*}
where $\Phi_p^{\ast}(r, s; \mathcal{X}_E)$ is given by (\ref{set2}). 
\end{enumerate}
\end{theorem}

\begin{proof}
Our proof follows the same line as the proof of Theorem 1 of \cite{sommerfeld2018inference}. Under the assumption of the theorem, the central limit theorem implies
\[
\sqrt{\frac{nm}{n+m}} \{ (\hat{r}_n, \hat{s}_m) - (r, s) \}
\stackrel{d}{\to}
(\sqrt{\delta}G, \sqrt{1-\delta}H),
\]
as $n \land m \to \infty$. 

\textit{About (i)}:
An application of the delta method in Theorem \ref{delta_method} with the directional Hadamard derivative of the map $(r, s) \mapsto \mathrm{IW}_p^p(r, s)$ given in Proposition \ref{derive_IPRW} yields
    \begin{equation}
        \sqrt{\frac{nm}{n+m}} \mathrm{IW}_p^p(\hat{r}_n, \hat{s}_m)
        \stackrel{d}{\to}
        \int_{S_{d, k}}
\max_{(u, v) \in \Phi_p^{\ast}(r, s; \mathcal{X}_E)} 
\langle u, \sqrt{\delta}G \rangle 
+  \langle v, \sqrt{1-\delta}H \rangle  
d\mu(E).
\label{IW_delta}
\end{equation}
Note that, under $r = s$,  we have $(u, v) \in \Phi^\ast(r, s; \mathcal{X}_E)$ if and only if $u \in \Phi^\ast_p(\mathcal{X}_E)$ and $v = -u$. Therefore, with $G \stackrel{d}{=} H$, we have
\begin{align}
    \max_{(u, v) \in \Phi_p^{\ast}(r, s; \mathcal{X}_E)} 
\langle u, \sqrt{\delta}G \rangle 
+  \langle v, \sqrt{1-\delta}H \rangle 
& \stackrel{d}{=}
\max_{(u, v) \in \Phi_p^{\ast}( \mathcal{X}_E)}
  \sqrt{\delta} \langle G, u \rangle 
- \sqrt{1-\delta} \langle H, u \rangle \notag  \\
&\stackrel{d}{=}
\max_{(u, v) \in \Phi_p^{\ast}(\mathcal{X}_E)}
\sqrt{\delta + (1 - \delta)}
 \langle G, u \rangle \notag  \\
&=
 \max_{(u, v) \in \Phi_p^{\ast}( \mathcal{X}_E)}
 \langle G, u \rangle,
 \label{limit_representation}
\end{align}
for each $E \in S_{d, k}$. (\ref{IW_delta}), (\ref{limit_representation}), and the application of the continuous mapping theorem with a map $t \mapsto t^{1/p}$ provide the conclusion.

\textit{About (ii)}:
Consider a map $(r, s) \mapsto \mathrm{IW}_p(r, s) = (\mathrm{IW}_p^p(r, s))^{1/p}$. By Proposition \ref{derive_IPRW} and the chain rule for directional Hadamard derivatives (Proposition 3.6 of \cite{shapiro1990concepts}), the directional Hadamard derivative of this map at $(r, s)$ is given by
\[
(h_1, h_2)
\mapsto
\frac{1}{p}\mathrm{IW}_p^{1-p}(r, s)
 \int_{S_{d, k}}
\max_{(u, v) \in \Phi_p^{\ast}(r, s; \mathcal{X}_E)} 
\langle u, h_1 \rangle + \langle v, h_2 \rangle  
d\mu(E).
\]
An application of the delta method in Theorem \ref{delta_method} yields the conclusion.
\end{proof}

The scaling rate in Theorem \ref{asymdist1} is independent of the dimension of underlying space $\mathcal{X}$, which is the same as the other extensions of the Wasserstein distance on finite spaces \citep{sommerfeld2018inference, klatt2020empirical, bigot2019central}.
Moreover,
for $p \ge 2$, the scaling rate in the case of $r = s$ (i.e., $n^{-1/2p}$) is slower than that in the case of $ r \neq s$ (i.e., $n^{-1/2}$). This implies that $\mathrm{IW}_p(\hat{r}_n, \hat{s}_m)$ converges faster under $r = s$ for $p \ge 2$. 

\subsection{Distributional limit for regularized PRW distance}
As our second main result, we derive a distributional limit of the empirical regularized PRW distance, $\mathrm{PW}_{p,\lambda}(\hat{r}_n, \hat{s}_m)$.
To study the PRW distance, we need to introduce the entropic regularization to add smoothness to the Wasserstein distance. For the regularization of the Wasserstein distance in a finite space, we refer to \cite{klatt2020empirical}.

We derive a distributional limit by showing the directional Hadamard differentiability of the regularized PRW distance and apply the delta method.
Our proof relies on the results of the sensitivity analysis in nonlinear programming \citep{fiacco1983introduction}, which is as follows.

Let us consider the following optimization problem containing a parameter $u \in U$ in the objective function:
\[
\max_{x \in \mathbb{R}^n} f(x, u) \quad \text{subject to} \quad x \in S, 
\]
where $f: \mathbb{R}^n \times U \to \mathbb{R}$, and $\nabla_uf$ is continuous on $\mathbb{R}^n \times U$. Moreover, feasible region $S \subset \mathbb{R}^n$ is a nonempty closed set and parameter set $U \subset \mathbb{R}^p$ is open and bounded. We define the optimal value function $\phi: U \to \mathbb{R}$ and the optimal set mapping $\Phi: U \to \mathcal{P}(\mathbb{R}^n)$ as:
$
\phi(u) = \max\{ f(x, u): x \in S \}
$
and
$
\Phi(u) = \{x \in S: \phi(u) = f(x, u) \}.
$
Then, we have the following result.

\begin{theorem}[Theorem 2.3.1 in \cite{fiacco1983introduction}]
For all $u \in U$ and in any direction $h \in \mathbb{R}^p$, the optimal value function $\phi$ is directionally differentiable in the sense of
 G\^{a}teaux; that is, limit (\ref{limit}) exists for a fixed $h$ and not a sequence $h_n \to h$. Additionally, the derivative is given by
 \[
 h \mapsto \max_{x \in \Phi(u)} \langle \nabla_uf(x, u), h \rangle.
 \]
 \label{sensitivity}
\end{theorem}
We employ this result to demonstrate the directional Hadamard differentiability of the regularized PRW distance.

For technical reasons, we reformulate regularized optimal transport problem (\ref{EOT}). 
The transport condition in (\ref{transport}) can be stated in terms of only the $2N - 1$ equality constraints instead of $2N$, which allows for linearly independent constraints. Following \cite{klatt2020empirical}, we denote by $A_{\star}$ and $s_{\star}$ the deletion of the last row of matrix $A$ in (\ref{transport}) and the last entry of vector $s \in \Delta_N$, respectively. We denote the set of such $s_\star$ as $(\Delta_N)_{\star}$. Using constraint $\sum_{i=1}^N s_i = 1$, we can identify vector $s \in \Delta_N$ with $s_\star \in (\Delta_N)_\star$. To apply Theorem \ref{sensitivity} to the regularized PRW distance, we show the continuous differentiability of the regularized optimal transport plan with projection in the following lemma.

\begin{lemma}
\label{difcont}
Let $p \ge 1$ be even and $\lambda > 0$. The map $(r, s_{\star}, E) \mapsto \pi_{p, \lambda}(r, s_{\star}; \mathcal{X}_E)$ is continuously differentiable 
on $\Delta_N \times (\Delta_N)_\star \times \mathbb{R}^{d \times k}$. In addition, the matrix of partial derivatives with respect to $(r, s_\star)$ at $(r_0, (s_0)_{\star}, E_0)$ is given by
\[
\nabla_{(r, s_{\star})} \pi_{p, \lambda}(r_0, s_{0\star}; \mathcal{X}_{E_0}) = D A_{\star}^\top  (A_\star D A_\star^\top )^{-1} \in \mathbb{R}^{N^2 \times (2N-1)},
\]
where $D \in \mathbb{R}^{N^2 \times N^2}$ is a
diagonal matrix whose $(j,j)$-entry is a $j$-th element of $\pi_{p, \lambda}(r_0, s_{0 \star}; \mathcal{X}_{E_0})$.
\end{lemma}

\begin{proof}
Our proof is similar to the proof of Theorem 2.3 in \cite{klatt2020empirical}, which shows the continuous differentiability of a regularized optimal transport plan without projection.
Note that regularized optimal transport \eqref{EOT} with marginal $r_0$ and $s_0$ satisfies the Slater's constraint qualification (Proposition 26.18 in \cite{bauschke2011convex}). Therefore, strong duality holds and the dual problem admits an optimal solution.
In addition, we can characterize the regularized optimal transport plan $\pi_{p, \lambda}$ and its corresponding optimal dual solution $\mu_{p, \lambda} \in \mathbb{R}^{2N-1}$ by the necessary and sufficient Karush-Kuhn-Tucker conditions: 
\[
c_p(\mathcal{X}_E) + \lambda \nabla \phi(\pi_{p, \lambda})^\top  - A_{\star}^\top  \mu_{p, \lambda} = 0, 
\quad
A_\star \pi_{p, \lambda} - (r_0, s_{0 \star})^\top  = 0.
\]
We now obtain the statement  by applying the implicit function theorem to this system of equations. Let us define a function
$F: \mathbb{R}^{N^2} \times \mathbb{R}^{2N - 1} \times \mathbb{R}^{2N-1} \times \mathbb{R}^{d \times k} \to \mathbb{R}^{N^2 + 2N - 1}$
 by 
\[
F(\pi, \mu, (r, s_{\star}), E)
=
\left(
  \begin{array}{c}
     c_p(\mathcal{X}_E) + \lambda \nabla \phi(\pi)^\top  - A_\star^\top  \mu \\
     A_\star \pi - (r, s_\star)^\top 
  \end{array}
 \right).
\]
Because $p$ is even, $F$ is continuously differentiable in the neighborhood of the specific point $(\pi_{p, \lambda}, \mu_{p, \lambda}, (r_0, s_{0\star}), E_0)$ with $F(\pi_{p, \lambda}, \mu_{p, \lambda}, (r_0, s_{0\star}), E_0) = 0$.
The matrix of the partial derivatives of $F$ with respect to $\pi$ and $\mu$ is given by
\[
\nabla_{(\pi, \mu)} F(\pi_{p, \lambda}, \mu_{p, \lambda}, (r_0, s_{0\star}), E_0)
=
\begin{pmatrix}
\lambda \nabla^2 \phi(\pi_{p, \lambda}) & -A_\star^\top \\
A_\star & 0 \\
\end{pmatrix} \in \mathbb{R}^{(N^2+2N-1) \times(N^2+2N-1)  }.
\]
This matrix is non-singular because $\lambda > 0$, the Hessian $\nabla \phi(\pi_{p, \lambda})$ is positive definite (Section 2.1 in \cite{klatt2020empirical}]) and the matrix $A_\star^\top$ has full rank.
As a result, the implicit function theorem guarantees the existence of a continuously differentiable function that parameterizes the regularized optimal transport plan with projection in the neighborhood of $(r_0, s_{0\star}, E_0)$.
The computation of the partial derivative form is directly followed by 
Theorem 2.3 and Example 2.6 in \cite{klatt2020empirical}.
\end{proof}

Now, we show the directionally Hadamard differentiability of the  regularized PRW distance. 
Given $(r, s_\star) \in \Delta_N \times (\Delta_N)_\star$, we define $\Psi_{p}^{\ast}(r, s_\star)$ as the set of directions that maximizes the regularized optimal transport distance between the projections of $r$ and $s$, that is,
\[
\Psi_p^{\ast}(r, s_{\star}) 
=
\{E \in S_{d, k}: W_{p, \lambda}(r, s; \mathcal{X}_E)
=
\mathrm{PW}_{p, \lambda}(r, s)
\}.
\]
We denote by $h_\star$ the deletion of the last entry of vector $h \in \Omega_N$ and the set of such $h_\star$ as $(\Omega_N)_\star$.
\begin{proposition}
Let $p \ge 1$ be even and let $\lambda > 0$. The map $(r, s_\star) \mapsto \mathrm{PW}_{p, \lambda}(r, s_\star)$ is directionally Hadamard differentiable at all $(r, s_\star) \in \Delta_N \times (\Delta_N)_{\star}$ tangentially to $\Omega_N \times (\Omega_N)_{\star}$ with the following derivative:
\begin{equation}
    (h_1, h_{2 \star}) \mapsto \max_{E \in \Psi_p^{\ast}(r, s_{\star})}
\langle \gamma^\top  D A_\star^\top  (A_\star D A_\star^\top )^{-1}, (h_1, h_{2\star})
\rangle, 
\label{deriv}
\end{equation}
where
\begin{equation}
\gamma = \frac{1}{p} \langle c_p(\mathcal{X}_E), \pi_{p, \lambda}(r, s_\star, \mathcal{X}_{E})\rangle^{\frac{1}{p}-1} c_p(\mathcal{X}_E)
\in \mathbb{R}^{N^2},
\label{gamma}
\end{equation}
and $D \in \mathbb{R}^{N^2 \times N^2}$ is a
diagonal matrix whose $(j,j)$-entry is a $j$-th element of $\pi_{p, \lambda}(r, s_\star, \mathcal{X}_E)$ for $j = 1,...,N^2$.
\label{derive_PRW}
\end{proposition}
\begin{proof}
Since the regularized optimal transport distance is defined as
\[
W_{p, \lambda}(r, s_\star; \mathcal{X}_E) = \langle c_p(\mathcal{X}_E), \pi_{p, \lambda}(r, s_\star; \mathcal{X}_E) \rangle^{1/p},
\]
it follows from Lemma \ref{difcont} that 
the map $(r, s_\star, E) \mapsto W_{p, \lambda}(r, s_\star, \mathcal{X}_E)$ can be continuously differentiated on $\Delta_N \times (\Delta_N)_\star \times \mathbb{R}^{d \times k}$.
 Moreover, the matrix of partial derivatives with respect to $(r, s_\ast)$ is given by
\[
\nabla_{(r, s_\star)} W_{p, \lambda}(r, s_\star; \mathcal{X}_E)
=
\gamma^\top  D A_\star^\top  (A_\star D A_\star^\top )^{-1},
\]
where $\gamma$ is the gradient of function $\pi \mapsto \langle c_p(\mathcal{X}_E), \pi \rangle^{1/p}$ evaluated in the regularized transport plan $\pi_{p, \lambda}(r, s_\star; \mathcal{X}_E)$, which is formally defined by (\ref{gamma}).
Consequently, Theorem \ref{sensitivity} implies that the map $(r, s_\star) \mapsto \mathrm{PW}_{p, \lambda}(r, s_\star)$ is directionally differentiable with derivative (\ref{deriv}) in the sense of G\^{a}teaux. 
To see this is also a directionally derivative in the Hadamard sense, 
it is sufficient to show the local Lipschitz continuity of this map (Proposition 3.5 of \cite{shapiro1990concepts}). 
To this end, we fix a closed set $S_0 \subset \Delta_N \times (\Delta_N)_\star$. For any $(r, s_\star), (r', s'_\star) \in S_0$, we have
\begin{align}
\label{maxineq}
    |\mathrm{PW}_{p, \lambda}(r, s_\star) - \mathrm{PW}_{p, \lambda}(r', s'_\star)|
    \le
    \max_{E \in S_{d, k}} |W_{p, \lambda}(r, s_\star; \mathcal{X}_E) -
    W_{p, \lambda}(r', s'_\star; \mathcal{X}_E)|.
\end{align}
Because the map $(r, s_\star, E) \mapsto W_{p, \lambda}(r, s_\star, \mathcal{X}_E)$ is continuously differentiable, 
there exists a constant $C > 0$ that does not depend on $(r, s_\star), (r', s'_\star)$ or $E$ so that
\begin{equation}
\label{lip1}
    |W_{p, \lambda}(r, s_\star; \mathcal{X}_E) -
    W_{p, \lambda}(r', s'_\star; \mathcal{X}_E)|
\le
C \|(r, s_\star) - (r', s'_\star) \|.
\end{equation}
A combination of equations (\ref{maxineq}) and  (\ref{lip1}) leads to local Lipschitz continuity of the map  $(r, s_\star) \mapsto \mathrm{PW}_{p, \lambda}(r, s_\star)$.
This completes the proof.
\end{proof}

The next theorem states our main result regarding the limit distribution of the empirically regularized PRW distance. 
\begin{theorem}[Distributional limit of $\mathrm{PW}_{p, \lambda}(\hat{r}_n, \hat{s}_m)$]
\label{limitdist2}
Let $p \ge 1$ be even and $\lambda > 0$. Under the assumptions of Theorem \ref{asymdist1}, as $n \land m \to \infty$ and $m/(n+m) \to \delta \in (0, 1)$, we have
\begin{align*}
    &\sqrt{\frac{nm}{n+m}} \{ \mathrm{PW}_{p, \lambda}(\hat{r}_n, \hat{s}_m)  - \mathrm{PW}_{p, \lambda}(r, s)\}  \\
    &\stackrel{d}{\to} 
    \max_{E \in \Psi_p^{\ast}(r, s_\star)}
    \langle \gamma^\top  D A_\star^\top  (A_\star D A_\star^\top )^{-1}, (\sqrt{\delta} G, \sqrt{1-\delta}H_\star) \rangle,
\end{align*}
where $\gamma \in \mathbb{R}^{N^2}$ and $D \in \mathbb{R}^{N^2 \times N^2}$ are defined in
Proposition \ref{derive_PRW} and $H_\star$ denotes the deletion of the last entry of random vector $H \sim N(0, \Sigma(s))$.
\end{theorem}

\begin{proof}
The proof is a simple application of the delta method (Theorem \ref{delta_method}), with the derivative of the regularized PRW distance (Proposition \ref{derive_PRW}).
\end{proof}

\section{Bootstrap} \label{sec:bootstrap}

We consider approximating the derived limit distributions by a bootstrap procedure.
Let $r, s \in \Delta_N$ and $X_1, ..., X_n \sim r, Y_1, ..., Y_m \sim s$ be i.i.d. samples with empirical distributions $\hat{r}_n$ and $\hat{s}_m$. 
Furthermore, let $\hat{r}_\ell^{\ast} $ and $\hat{s}_\ell^{\ast} $ be empirical bootstrap distributions defined using i.i.d. bootstrap samples $X_1^\ast, ..., X_\ell^\ast \sim \hat{r}_n$ and $Y_1^\ast, ..., Y_\ell^\ast \sim \hat{s}_m$.

The functionals $\mathrm{IW}_p$ and $\mathrm{PW}_{p, \lambda}$ are only directionally Hadamard differentiable, that is, they have nonlinear derivatives with respect to $(h_1, h_2)$. As mentioned by \cite{dumbgen1993nondifferentiable} and \cite{sommerfeld2018inference}, the naive $n$-out-$n$ bootstrap is inconsistent for such functionals with a nonlinear Hadamard derivative, but re-sampling fewer than $n$ observations leads to a consistent bootstrap (the rescaled or $m$-out-$n$ bootstrap).
From this fact, we obtain the following results regarding the bootstrap for the IPRW and regularized PRW distances.
In the following, $\text{BL}_1({\mathbb{R})}$ denotes the set of all bounded functions on $\mathbb{R}$ with a Lipschitz constant of at most one. 

\begin{proposition}
\label{bootsw}
Let $p \ge 1$.
We assume that $\ell \to \infty,\ell/n \to \infty$ and $\ell/m \to \infty$ as $n, m \to \infty$. Then, the plug-in bootstrap with $\hat{r}_\ell^\ast$ and $\hat{s}_\ell^\ast$ for the integral projection robust Wasserstein distance is consistent:
\begin{enumerate}
    \item  If $r = s$, and $n \land m \to \infty$ and $m/(n+m) \to \delta \in (0, 1)$, we have:
\begin{align*}
    \sup_{h \in \text{BL}_1(\mathbb{R})} 
    \Biggl| &\mathbb{E} \left[ \left.  h  \left(\left(\frac{\ell}{2}\right)^{\frac{1}{2p}}\mathrm{IW}_p(\hat{r}_\ell^\ast, \hat{s}_\ell^\ast)\right) \right| X_1, ..., X_n, Y_1, ..., Y_m \right]   \\
&-
\mathbb{E}\left[ h\left(\left(\frac{nm}{n+m}\right)^{\frac{1}{2p}}\mathrm{IW}_p(\hat{r}_n, \hat{s}_m)\right)\right] \Biggr|
\to
0,
\end{align*}
in outer probability.
\item If $r \neq s$ and $n \land m \to \infty$ and $m/(n+m) \to \delta \in (0, 1)$, we have:
\begin{align*}
\sup_{h \in \text{BL}_1(\mathbb{R})} \Biggl|&\mathbb{E}\left[ \left. h\left(\sqrt{\frac{\ell}{2}} \{\mathrm{IW}_p(\hat{r}_\ell^\ast, \hat{s}_\ell^\ast) - \mathrm{IW}_p(\hat{r}_n, \hat{s}_m)\}\right) \right| X_1, ..., X_n, Y_1, ..., Y_m \right] \\
&-
\mathbb{E}\left[ h\left(\sqrt{\frac{nm}{n+m}} \{\mathrm{IW}_p(\hat{r}_n, \hat{s}_m) - \mathrm{IW}_p(r,s)\}\right) \right]  \Biggr|
\to
0,
\end{align*}
in outer probability.
\end{enumerate}
\end{proposition}

\begin{proof}
As shown in Proposition \ref{derive_IPRW}, the map $(r, s) \mapsto \mathrm{IW}_p(r, s)$ is 
directionally Hadamard differentiable. 
Then,
the proof is a direct application of Proposition 2 in \cite{dumbgen1993nondifferentiable} with this map.
\end{proof}

\begin{proposition}
\label{bootprw}
Let $p \ge 1$ be even and $\lambda > 0$.
We assume that $\ell \to \infty, \ell/n \to \infty$ and $\ell/m \to \infty$ as $n, m \to \infty$. Then, the plug-in bootstrap with $\hat{r}_\ell^\ast$ and $\hat{s}_\ell^\ast$ for the regularized projection robust Wasserstein distance is consistent. That is, as $n \land m \to \infty$ and $m/(n+m) \to \delta \in (0, 1)$, we have
\begin{align*}
\sup_{h \in \text{BL}_1(\mathbb{R})} \Biggl| &\mathbb{E}\left[ \left. h\left(\sqrt{\frac{\ell}{2}} \{\mathrm{PW}_{p, \lambda}(\hat{r}_\ell^\ast, \hat{s}_\ell^\ast) - \mathrm{PW}_{p, \lambda}(\hat{r}_n, \hat{s}_m)\}\right) \right| X_1, ..., X_n, Y_1, ..., Y_m \right] \\
&-
\mathbb{E}\left[ h\left(\sqrt{\frac{nm}{n+m}} \{\mathrm{PW}_{p, \lambda}(\hat{r}_n, \hat{s}_m) - \mathrm{PW}_{p, \lambda}(r,s)\}\right) \right]  \Biggr|
{\to}
0,    
\end{align*}
in outer probability.
\end{proposition}

\begin{proof}
As shown in Proposition \ref{derive_PRW}, the map $(r, s) \mapsto \mathrm{PW}_{p, \lambda}(r, s)$ is 
directionally Hadamard differentiable. 
Then,
the proof is a direct application of Proposition 2 in \cite{dumbgen1993nondifferentiable} with this map.
\end{proof}

In practice, the performance of our bootstrap procedure depends on the choice of replacement number $\ell$. In Section 5, we investigate how the choice of $\ell$ affects the finite-sample performance of bootstrapping using simulation studies.

\section{Simulation studies}
We illustrate our distributional limit results in Monte Carlo simulations. Specifically, we investigate the speed of convergence for the empirical IPRW distance ($p=1$) and the empirical regularized PRW distance ($p=2$) to their limit distributions (Theorems \ref{asymdist1} and \ref{limitdist2}).
We also illustrate the accuracy of the approximation using the rescaled bootstrap (Propositions \ref{bootsw} and \ref{bootprw}). 
All simulations were performed using \textsf{R} (\cite{team2013r}). The Wasserstein distances were calculated using the \textsf{R}~package \textit{transport} \citep{r_transport}, and the regularized transport distances were calculated using the \textsf{R}~package \textit{Barycenter} \citep{r_barycenter}.

\subsection{Speed of convergence to the distributional limit}

\subsubsection{Integral projection robust Wasserstein distance}
We consider finite ground space $\mathcal{X}$ to be an equidistant two-dimensional $L \times L$ grid on $[0, 1] \times [0, 1]$, with size $N = L^2$. 
We first set the grid size to $L=7$ (i.e., $N=49$).

For case $r=s$, we consider a probability distribution $r$ on $\mathcal{X}$ as the realization of a Dirichlet random variable Dir(\textbf{1}) with concentration parameter $\mathbf{1} = (1, ..., 1) \in \mathbb{R}^{N}$ and set $s=r$. Given such distributions $r, s \in \Delta_N$, we sample observations $X_1, ..., X_n \sim r$ and $Y_1, ..., Y_m \sim s$ i.i.d. with sample size $n = m \in \{25, 50, 100, 1000, 5000\}$ and compute $\sqrt{\frac{n}{2}}\mathrm{IW}_1(\hat{r}_n, \hat{s}_n)$ with one-dimensional projection and the uniform measure, which corresponds to the sliced Wasserstein distance.
This process is repeated 20,000 times.
Similarly, we consider the same setup for $r \neq s$, where we generate a second distribution, $s \sim$ Dir(\textbf{1}), independently. 
We then compare the finite distributions with the theoretical limit distributions given by Theorem \ref{asymdist1}.

We demonstrate the results using kernel density estimators and the corresponding Q-Q plots in Figure \ref{sw_plots} (A) and (B). The limit distributions are good approximations of the finite sample distributions for a large sample size $(n=1000)$ in both cases $r=s$ and $r \neq s$.
We also observe that, under $r = s$, the limit law approximates the finite sample distribution quite well, even for a small sample size $(n=50)$.
In Figure \ref{sw_ks}, 
we also show the speed of convergence with respect to the Kolmogorov--Smirnov distance (maximum absolute difference between the distribution function of the finite sample law and that of the limit law) for grid sizes $L=3, 5, 7$.
This shows that the Kolmogorov--Smirnov distances decrease as the sample size increases, and the size of ground space $N=L^2$ slows the speed of convergence marginally, especially for $r \neq s$.

\begin{figure}[htbp]
  \begin{minipage}[b]{0.9\linewidth}
    \centering
    \includegraphics[keepaspectratio, width=110mm]{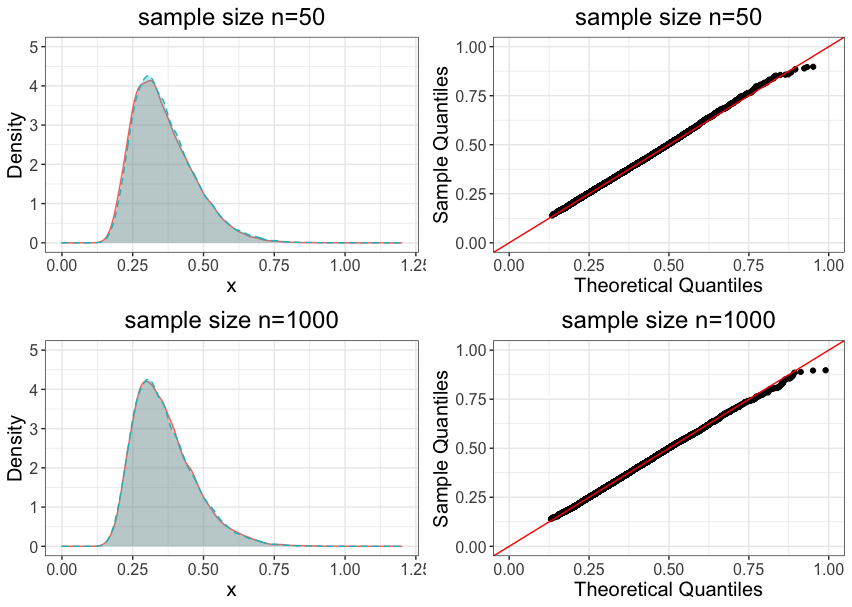}
    \subcaption{ $r=s$}
  \end{minipage} \\
  \begin{minipage}[b]{0.9\linewidth}
    \centering
    \includegraphics[keepaspectratio, width=110mm]{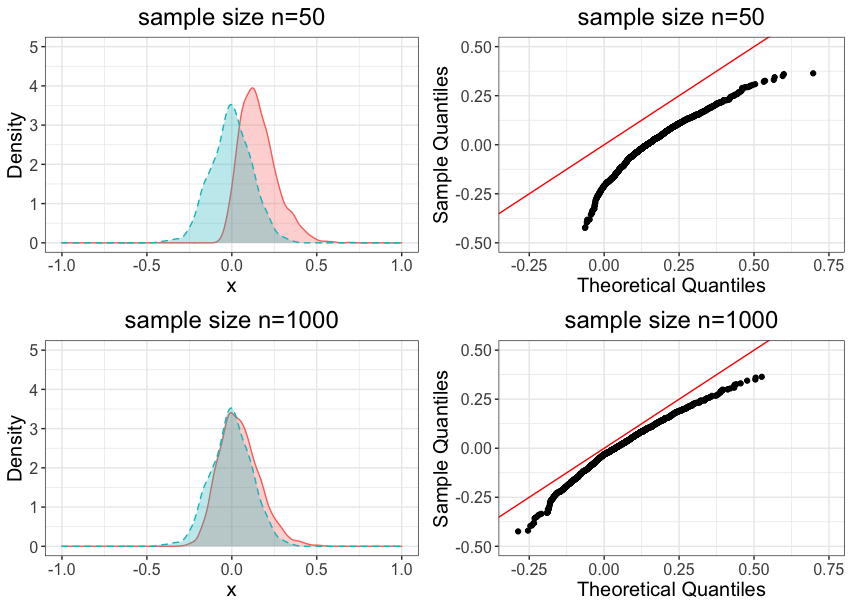}
    \subcaption{ $r \neq s$}
  \end{minipage}
      \caption{ \small{
      \textbf{(A) Comparison of the finite sample distributions and the limit distribution of the empirical IPRW distance for the case $\mathbf{r=s}$.} First row shows finite sample density (dashed line) of the empirical IPRW distance for $n=50$ on a regular grid of size $L=7$ compared to its limit density (solid line). The densities are estimated by kernel density estimators with Gaussian kernel and Silverman's rule is used  to select bandwidth. The corresponding Q-Q plot is presented on the right, where the red solid line indicates perfect fit. Second row is the same setting as above, but $n=1000$. 
      \textbf{(B) Comparison of the finite sample distributions and the limit distribution of the empirical IPRW distance for the case $\mathbf{r \neq s}$.} Same scenario as in (A), but here the sampling distributions $r$ and $s$ are not equal.}}
  \label{sw_plots}
\end{figure}

\begin{figure}[htbp]
  \begin{minipage}[b]{0.45\linewidth}
    \centering
    \includegraphics[keepaspectratio, width=70mm]{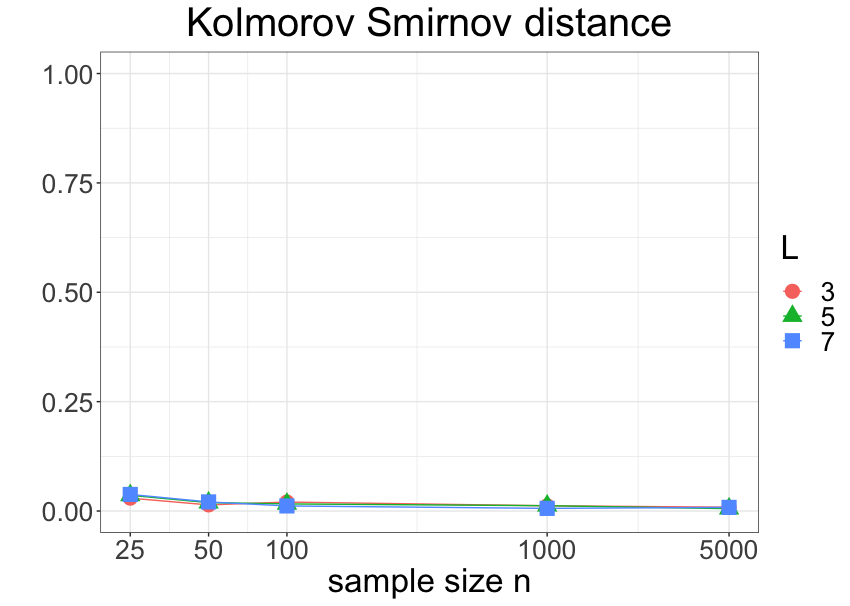}
    \subcaption{ $r=s$}
  \end{minipage}
  \begin{minipage}[b]{0.45\linewidth}
    \centering
    \includegraphics[keepaspectratio, width=70mm]{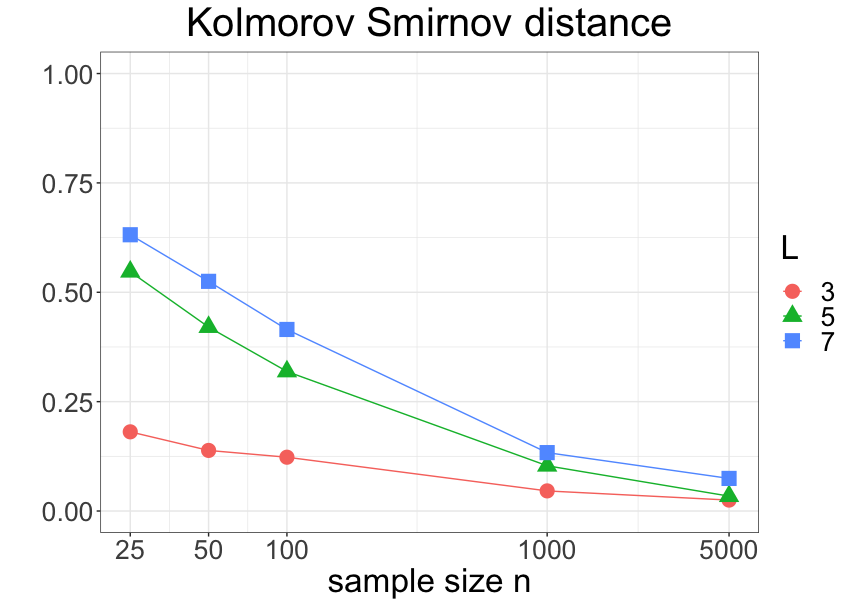}
    \subcaption{ $r \neq s$}
  \end{minipage}
      \caption{\small{
       \textbf{(A) Kolmogorov-Smirnov distances of the IPRW distance for the case $\mathbf{r=s}$.} The Kolmogorov-Smirnov distance between the finite sample distributions of the empirical IPRW distance and its theoretical limit distribution for different sample size $n \in \{25, 50, 100, 1000, 5000\}$ and different grid sizes $L$. The axes are given on a logarithmic scale.
       \textbf{
       (B) Kolmogorov-Smirnov distances of the IPRW distance for the case $\mathbf{r \neq s}$.} Same scenario as in (A), but here the sampling distributions $r$ and $s$ are not equal.
      }}
      \label{sw_ks}
\end{figure}

\subsubsection{Regularized projection robust Wasserstein distance}
We consider ground space $\mathcal{X}$ to be of form $\{1/M, 2/M, ..., M/M\} \times \{-0.001, 0.001 \} \times\{-0.001, 0.001\} \subset \mathbb{R}^3$ with grid size $M$ and total size $N=4M$. This ground space $\mathcal{X}$ is set to have a low-dimensional structure: two distributions on $\mathcal{X}$ differ mostly in the first coordinate, while the differences in the second and third coordinates are regarded as noise. For $M=10$ (i.e., $N=40$), we generated probability distributions $r$ and $s$ on $\mathcal{X}$ as realizations of independent Dirichlet random variables Dir(\textbf{1}). Given distributions $r \neq s$, we consider the same sampling scenarios as in the case of the IPRW distance and compute $\sqrt{\frac{n}{2}} \{ \mathrm{PW}_{2, \lambda}(\hat{r}_n, \hat{s}_n) - \mathrm{PW}_{2, \lambda}(r, s) \}$ with one-dimensional projection and regularization parameter $\lambda = 1$. We repeat this process 20,000 times and compare the finite distribution to its theoretical limit distribution given by Theorem \ref{limitdist2}.

Figure \ref{prw_density} shows the results demonstrated by the kernel density estimators and corresponding Q-Q plots. The limit distributions are good approximations of the finite sample distributions for both small and large sample sizes. Figure \ref{prw_ks} shows the speed of convergence with respect to the Kolmogorov-Smirnov distance under grid sizes $M=3, 7, 10$. We observe a declining tendency of the Kolmogorov-Smirnov distances as the sample size increases.

\begin{figure}[htbp]
    \centering
    \includegraphics[keepaspectratio, width=110mm]{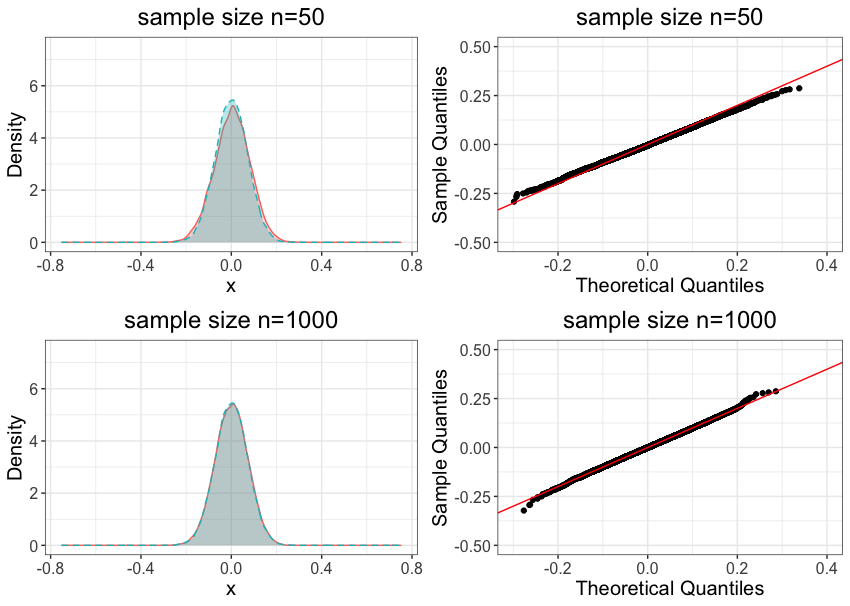}
    \caption{\small{\textbf{Comparison of the finite sample distributions and the limit distribution of the empirical  regularized PRW distance for the case $\mathbf{r\neq s}$}.  First row shows finite sample density (dashed line) of the empirical regularized PRW distance for $n=50$ on a ground space of grid size $M=10$ compared to its limit density (solid line). The densities are estimated in the same way as Figure \ref{sw_plots}. The corresponding Q-Q plot is presented on the right, where the red solid line indicates perfect fit. Second row is the same setting as above, but $n=1000$. }}
    \label{prw_density}
\end{figure}

\begin{figure}[htbp]
    \centering
    \includegraphics[keepaspectratio, width=70mm]{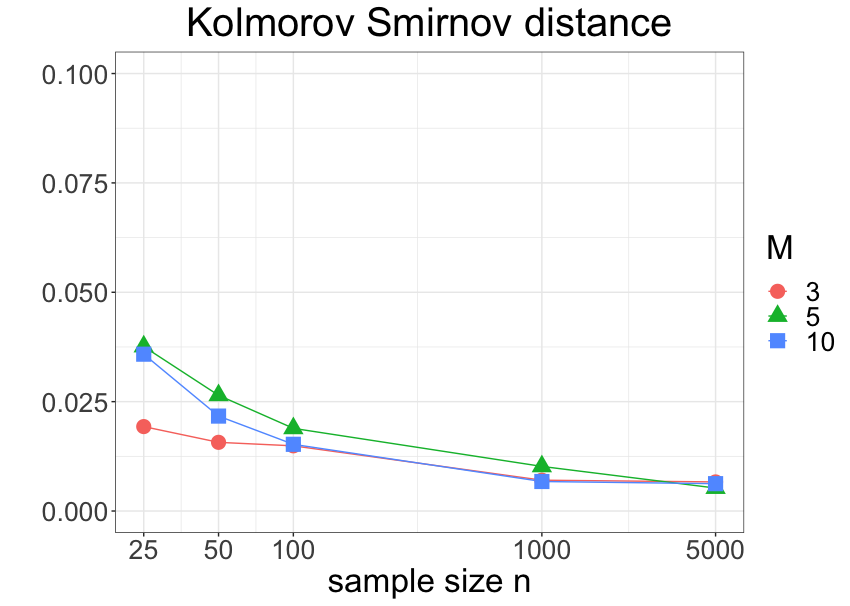}
    \caption{\small{\textbf{Kolmogorov-Smirnov distance of the regularized PRW for $\mathbf{r \neq s}$.} The Kolmogorov-Smirnov distance between the finite sample distribution of the empirical regularized PRW distance and its theoretical limit distribution for different sample size $n \in \{25, 50, 100, 1000, 5000\}$ and different grid sizes $M$. The axes are given on a logarithmic scale.}}
    \label{prw_ks}
\end{figure}

\subsection{Simulation of bootstrap}

\subsubsection{Integral projection robust Wasserstein distance}
We simulate the rescaled plug-in bootstrap approximations from Section \ref{sec:bootstrap} for the IPRW distance. For a grid with $L=7$, we generate $r \sim \text{Dir(\textbf{1})}$, set $s=r$, and sample $n=1000$ observations according to probability distributions $r, s$. In addition, for fixed empirical distributions $\hat{r}_n, \hat{s}_n$, we generate $B=500$ bootstrap replications of
$
\sqrt{\frac{\ell}{2}} \mathrm{IW}_1(\hat{r}_\ell^\ast, \hat{s}_\ell^\ast)
$
by drawing independently with replacement $ \ell \in \{n, n^{4/5}, n^{2/3}, n^{1/2} \}$ according to $\hat{r}_n$ and $\hat{s}_n$.
Similarly, we consider the same setup in the case of $r \neq s$, where the second distribution, $s$, is generated independently from $\text{Dir(\textbf{1})}$. In the $r \neq s$ case, the form of bootstrap replications is
$
\sqrt{\frac{\ell}{2}} 
\{\mathrm{IW}_1(\hat{r}_\ell^\ast, \hat{s}_\ell^\ast) - \mathrm{IW}_1(\hat{r}_n, \hat{s}_n)\}.
$
The finite bootstrap sample distributions are then compared with their finite sample and theoretical limit distributions.

\begin{figure}[htbp]
  \begin{minipage}[b]{0.9\linewidth}
    \centering
    \includegraphics[keepaspectratio, width=110mm]{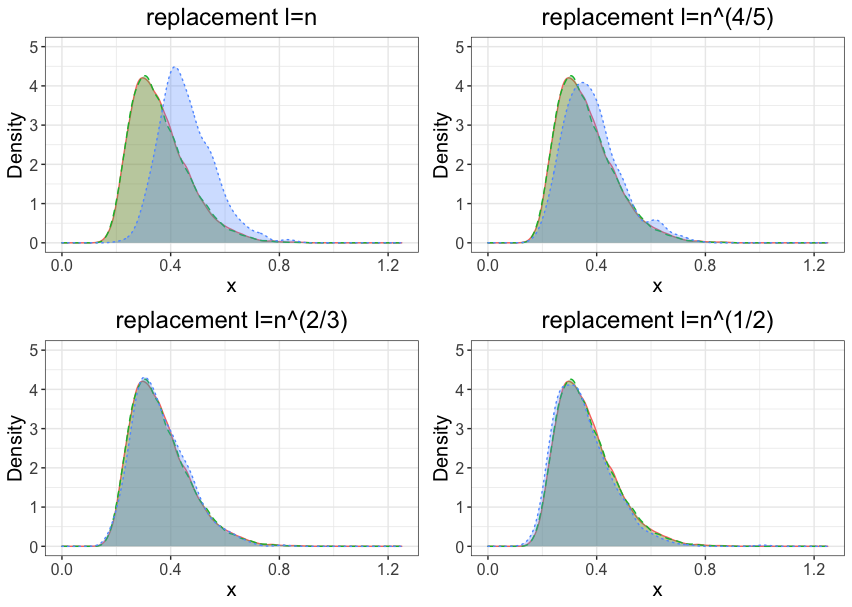}
    \subcaption{ $r=s$}
  \end{minipage} \\
  \begin{minipage}[b]{0.9\linewidth}
    \centering
    \includegraphics[keepaspectratio, width=110mm]{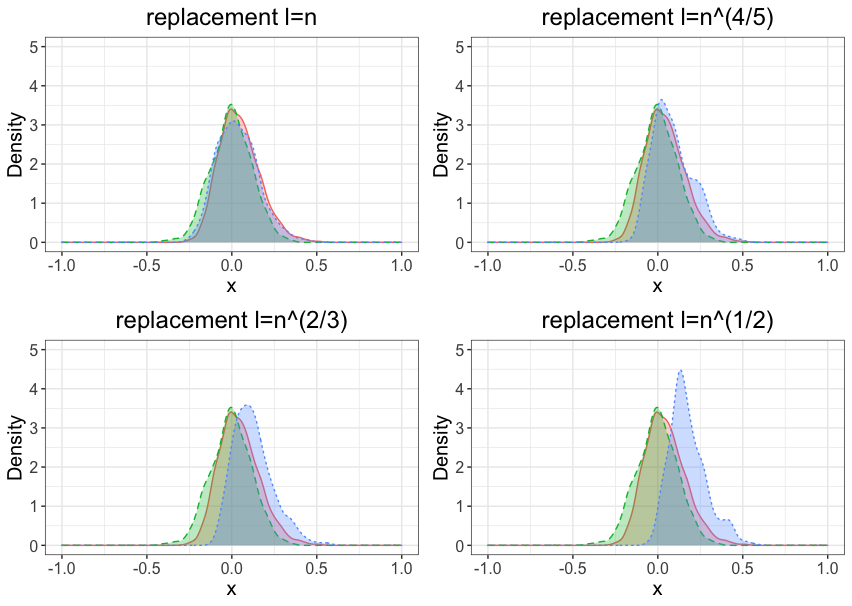}
    \subcaption{ $r \neq s$}
  \end{minipage}
      \caption{\small{\textbf{(A) Bootstrap for the empirical IPRW distance under $\mathbf{r = s}$.} Illustration of the rescaled plug-in bootstrap approximation $(n=1000)$ with replacement $\ell \in \{n, n^{4/5}, n^{2/3}, n^{1/2}\}$ and grid size $L=7$. Finite bootstrap densities (dotted lines) are compared to their finite sample density (solid line) and limit density (dashed line). 
      The densities are estimated in the same way with Figure \ref{sw_plots}.
      \textbf{(B) Bootstrap for the empirical IPRW distance under $\mathbf{r \neq s}$.} Same scenario as in (A), but here the sampling distributions $r$ and $s$ are not equal.}}
      \label{boot_sw}
\end{figure}

The results are shown in Figure \ref{boot_sw}. We observe that, under $r=s$, finite bootstrap distributions with fewer replacements $(\ell=n^{4/5}, n^{2/3}, n^{1/2})$ are better approximations of the finite sample distribution than the naive bootstrap ($\ell=n$). This is consistent with the theoretical result in Section 4, which claims the naive bootstrap does not have consistency for the IPRW distance but resampling fewer observations leads to consistency. However, under $r \neq s$, the bootstrap approximations with fewer replacements are not good, and the naive bootstrap approximation is better. 
This good approximation by the naive bootstrap is possible due to the fact that the map $(r, s) \mapsto \mathrm{IW}_p(r, s)$ is only directionally Hadamard differentiable in general but (non-directionally) Hadamard differentiable at most points $(r, s)$ with $r \neq s$. For instance, for ground size $N=2$ (i.e., $\mathcal{X} = \{x_1, x_2\}$), the IPRW distance can be explicitly written as $\mathrm{IW}_p(r, s) = (\int_{S_{d, k}} \|E^\top(x_1 - x_2)\|^p d\mu(E))^{1/p}|r_1 - s_1|$. Therefore, in this case, the map $(r, s) \mapsto\mathrm{IW}_p(r, s)$ is Hadamard differentiable if $r \neq s$.

\subsubsection{Regularized projection robust Wasserstein distance}
For grid size $M=10$, we generate distributions $r$ and $s$ as realizations of independent random variables from Dir($\mathbf{1}$) and sample $n=1000$ observations according to probability distributions $r, s$. Additionally, for fixed empirical distributions $\hat{r}_n, \hat{s}_n$, we generate $B=500$ bootstrap replications of
$
\sqrt{\frac{\ell}{2}} \{\mathrm{PW}_{2, \lambda}(\hat{r}_\ell^\ast, \hat{s}_\ell^\ast) -\mathrm{PW}_{2, \lambda}(\hat{r}_n, \hat{s}_n) \}
$
with $\lambda = 1$
by drawing independently with replacement $ \ell \in \{n, n^{4/5}, n^{2/3}, n^{1/2} \}$, according to $\hat{r}_n$ and $\hat{s}_n$. 
The finite bootstrap sample distributions are then compared with their finite sample and theoretical limit distributions. 

The results are shown in Figure \ref{boot_prw}. The accuracy of the bootstrap approximation is not affected by replacement number $\ell$ in this case.

\begin{figure}[htbp]
    \centering
    \includegraphics[keepaspectratio, width=110mm]{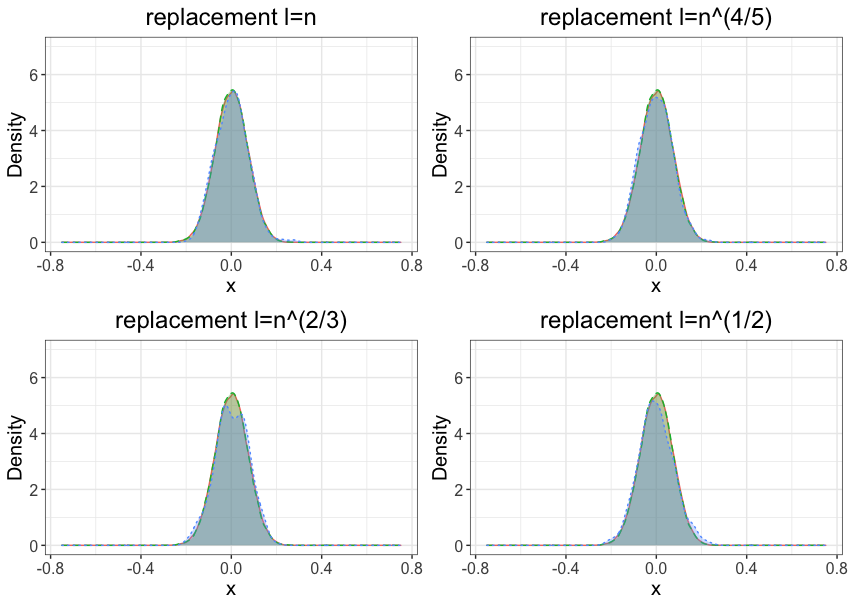}
    \caption{\small{\textbf{Bootstrap for the regularized PRW distance under $\mathbf{r \neq s}$}.  Illustration of the rescaled plug-in bootstrap approximation $(n=1000)$ with the replacement $\ell \in \{n, n^{4/5}, n^{2/3}, n^{1/2}\}$ and grid size  $M=10$. Finite bootstrap densities (dotted lines) are compared with their finite sample density (solid line) and limit density (dashed line). 
      The densities are estimated in the same way with Figure \ref{sw_plots}.}}
    \label{boot_prw}
\end{figure}

\section{Applications}
\subsection{Two-sample testing with sliced Wasserstein distance}
Let $r, s \in \Delta_N$ and take $X_1, ..., X_n \sim r$, $Y_1, ..., Y_m \sim s$ be i.i.d. samples. The nonparametric two-sample testing is a problem of detecting whether sampling distributions $r, s$ are equal, based on samples. This is described as
\[
H_0 : r = s \quad \text{vs}. \quad H_1: r \neq s.
\]
Based on the previous distributional results, we propose a test using the sliced Wasserstein distance, that is, the IPRW distance with one-dimensional projection and uniform measure.
Specifically, we denote $\mathrm{SW}_{m, n} = \sqrt{\frac{mn}{m+n}} \mathrm{IW}_p(\hat{r}_n, \hat{s}_m)$ and propose a test
\[
\mathrm{SW}_{m, n} > c_{\alpha} \Rightarrow
\,\,
\text{reject}
\,\,
H_0,
\]
where $c_\alpha$ is a critical value chosen according to the given level of $\alpha \in (0, 1)$. 
The two-sample testing based on the Wasserstein distance was performed by \cite{ramdas2017wasserstein}. They designed univariate test statistics using the Wasserstein distance and analyzed their limit distribution. However, their approach is only available for the $d=1$ case, as it does not extend to higher dimensions. Our proposed test is not restricted to a one-dimensional setting, being applicable to large-scale datasets because of the low computational complexity of the sliced Wasserstein distance.

We use the bootstrap procedure to choose an appropriate critical value from data. Let $\hat{r}_\ell^{\ast}$ and $\hat{s}_\ell^{\ast}$ be the empirical bootstrap distributions obtained from bootstrap samples $X_1^\ast, ..., X_\ell^\ast \sim \hat{r}_n$ and $Y_1^\ast, ..., Y_\ell^\ast \sim \hat{s}_m$, respectively. We define the bootstrap version of the test statistics as:
$
\mathrm{SW}_{m, n}^\ast = \sqrt{\frac{\ell}{2}} \mathrm{IW}_p(\hat{r}_\ell^\ast, \hat{s}_\ell^\ast)
$
and denote by $\hat{c}_\alpha$ the $(1-\alpha)$quantile of $\mathrm{SW}_{m,n}^{\ast}$. Note that $\hat{c}_\alpha$ can be computed numerically.
Then, the validity of the rescaled bootstrap for the IPRW distance (Proposition \ref{bootsw}) implies that, under  $\ell \to \infty, \ell/n \to 0$, and $\ell/m \to 0$ as $n,m \to \infty$, the test
\[
\mathrm{SW}_{m, n} > \hat{c}_{\alpha} \Rightarrow
\,\,
\text{reject}
\,\,
H_0
\]
has asymptotic level $\alpha$. Specifically, $\limsup_{m, n \to \infty} P(\mathrm{SW}_{m, n} > \hat{c}_{\alpha}) \le \alpha$.

We here illustrate the finite sample performance of this test. We set the finite ground space $\mathcal{X}$ to be an equidistant two-dimensional $7 \times 7$ grid on $[0, 1] \times [0, 1]$. For the case $r = s$, we generate a distribution $r \sim \text{Dir}(\textbf{1})$ and set $s=r$, while for the case $r \neq s$, we generate two distributions $r, s \sim \text{Dir}(\textbf{1})$ independently.
We set the sample size as $n=m=1000$ and vary the replacement number as $\ell \in \{n^{4/5}, n^{2/3}, n^{1/2}\}$. We set the significance level to be $\alpha = 0.05$ and run $1000$ Monte Carlo iterations in each case.

Table \ref{tab:test_power} shows the rejection rates of the proposed test in each case. For the case $r = s$, the rejection rates should be under the significance level $\alpha = 0.05$, and this is true for all $\ell \in \{n^{4/5}, n^{2/3}, n^{1/2}\}$. For the case $r \neq s$, the power of the test is $1.000$, which is satisfactory.

\begin{table}[]
    \centering
    \begin{tabular}{ccc}
         &$r=s$  & $r \neq s$  \\ \midrule
        $\ell = n^{4/5}$ &0.001 &1.000 \\
        $\ell = n^{2/3}$ &0.016 &1.000 \\
        $\ell = n^{1/2}$ &0.037 &1.000 \\
    \end{tabular}
    \caption{\small{Rejection rates of the proposed test. The significance level is $0.05$.}}
    \label{tab:test_power}
\end{table}

We now apply the proposed test to testing the equality of color distributions in images. Given two different images, the aim is to investigate whether the images have significantly different color distributions. Figure \ref{datasets} shows the datasets of images used. 
Each image has $768 \times 576 = 442368$ pixels.
We obtained these images from a publicly available dataset \url{http://tabby.vision.mcgill.ca/html/welcome.html}.
We transform each image into a color histogram in the RGB color space with grid size $16^3 = 4086$. In the dataset 1 (the first column in Figure \ref{datasets}), the two images are expected to have different color distributions.
In the dataset 2 (the second column in Figure \ref{datasets}), the two images are expected to have different but similar color distributions.
In the dataset 3 (the third row in Figure \ref{datasets}), one image is obtained by turning the other image from side to side; thus, they have the same color histograms.
In each dataset, we randomly select $n = 10,000$ pixels from each image and construct empirical color distributions $\hat{r}_n, \hat{s}_n$. We then
calculate the test statistics $\mathrm{SW}_{n,n}$ and $p$-values based on $B=500$ bootstrap with replacement $\ell \in \{n^{4/5}, n^{2/3}, n^{1/2} \}$. The results are shown in Table \ref{table_1}.

\begin{figure}[htbp]
    \begin{tabular}{ccc}
      \begin{minipage}[t]{0.3\hsize}
        \centering
        \includegraphics[keepaspectratio, width=40mm]{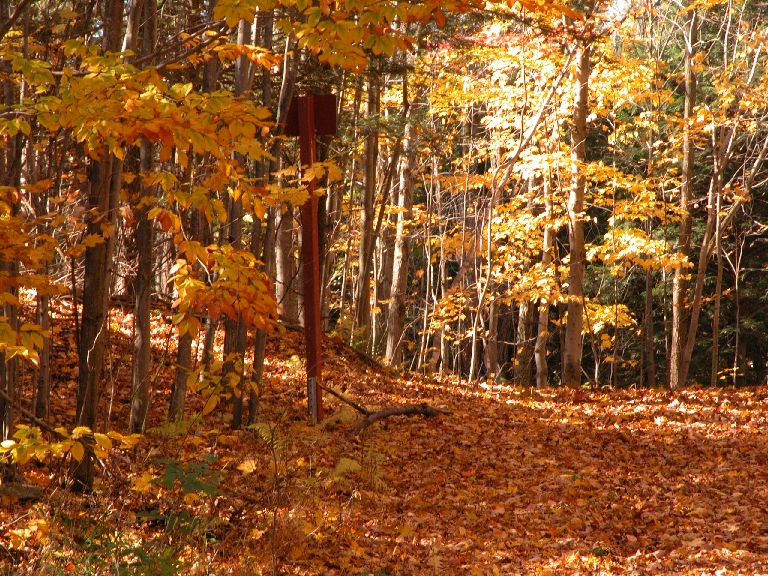}
      \end{minipage} &
      \begin{minipage}[t]{0.3\hsize}
        \centering
        \includegraphics[keepaspectratio, width=40mm]{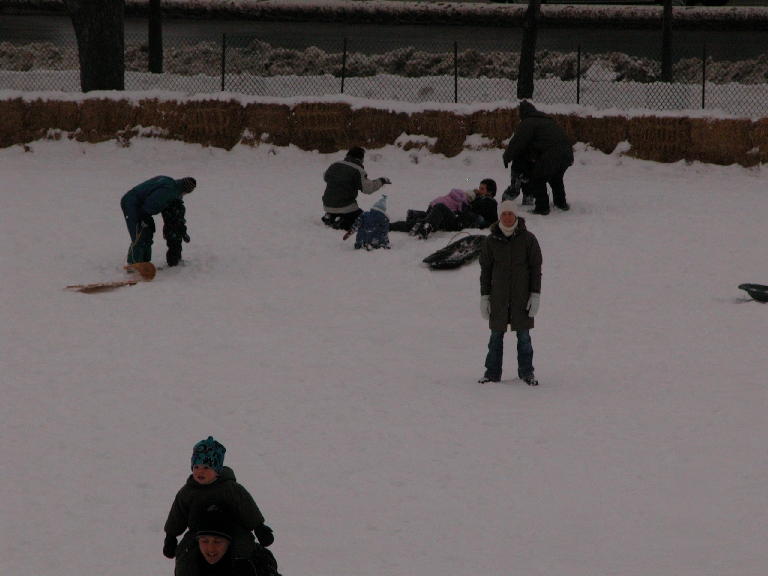}
      \end{minipage} &
      \begin{minipage}[t]{0.3\hsize}
        \centering
        \includegraphics[keepaspectratio, width=40mm]{figs_draft/samplemerry_mtl07_054.jpg}
      \end{minipage} \\
   
      \begin{minipage}[t]{0.3\hsize}
        \centering
        \includegraphics[keepaspectratio, width=40mm]{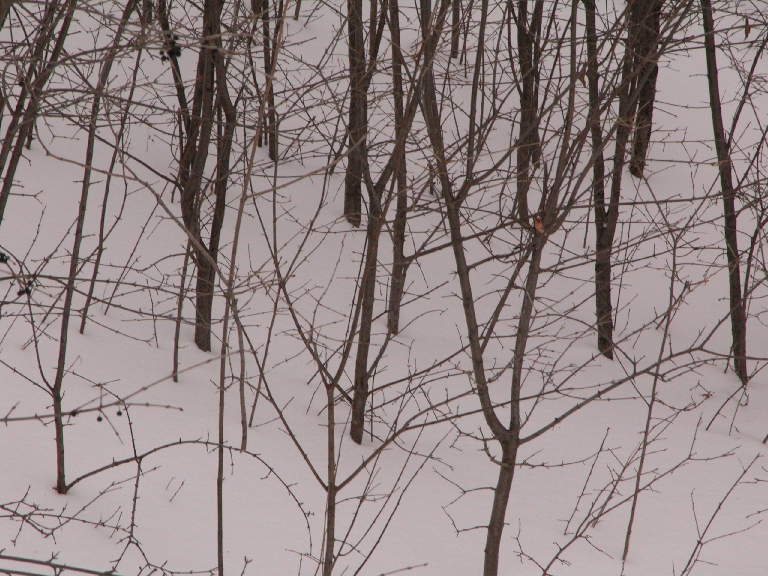}
      \end{minipage} &
      \begin{minipage}[t]{0.3\hsize}
        \centering
        \includegraphics[keepaspectratio, width=40mm]{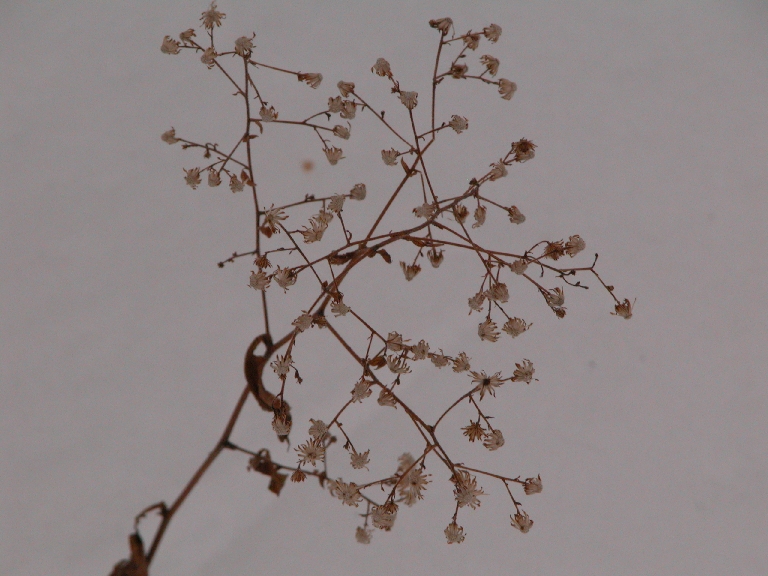}
      \end{minipage} &
      \begin{minipage}[t]{0.3\hsize}
        \centering
        \includegraphics[keepaspectratio, width=40mm]{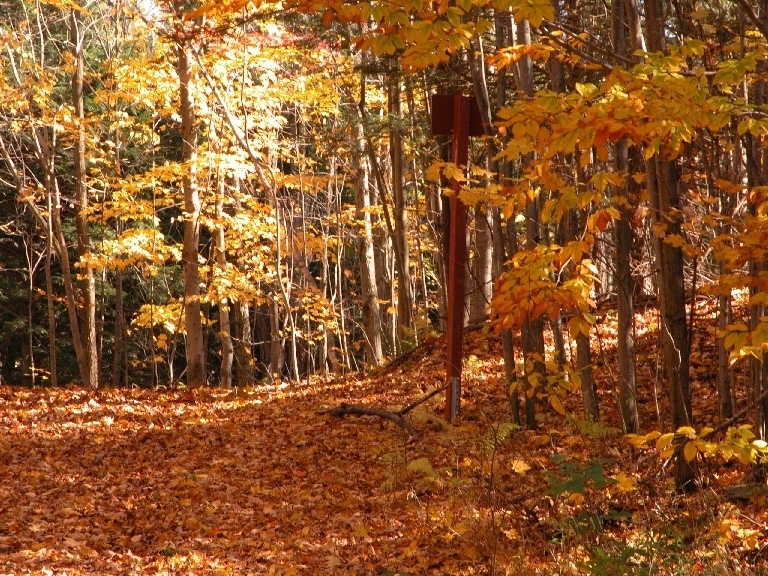}
      \end{minipage}
    \end{tabular}
    \caption{\small{Datasets of images. The first, second and third columns show the dataset 1, 2 and 3, respectively.}}
    \label{datasets}
  \end{figure}

\begin{table}[htbp]
    \centering
    \begin{tabular}{cccccc} 
    \toprule
    Dataset & Statistic & \multicolumn{3}{c}{$p$-value}  \\
     \cmidrule(lr){3-5}
    & &$\ell=n^{4/5}$&$\ell=n^{2/3}$&$\ell=n^{1/2}$ \\
    \midrule
    1 & 15.55 & $<$0.001 & $<$0.001 & $<$ 0.001 \\
    2 & 9.07 & $<$ 0.001 & $<$ 0.001 & $<$ 0.001 \\
    3 & 0.25 &0.446 & 0.372 & 0.352 \\
    \bottomrule
    \end{tabular}
    \caption{\small{Two-sample testing for the color distributions of the images}}
    \label{table_1}
\end{table}

We observe that, for the dataset 1, the proposed test with every replacement $\ell$ suggests a strong rejection of the null hypothesis.
For the dataset 2, we also see a strong rejection of the null hypothesis, but the test statistics (9.07) is smaller than that for the dataset 1 (15.55).
For the dataset 3, the proposed test with any replacement $\ell$ does not report a small $p$-value, which means there is no strong evidence to reject the null hypothesis. 

\subsection{Interval estimation for regularized projection robust Wasserstein distance}
Given a level $\alpha \in (0, 1)$ and i.i.d. samples $X_1, ..., X_n \sim r, Y_1, ..., Y_m \sim s$, we aim to construct an asymptotic confidence interval $C_{nm}$ for the regularized PRW distance $\mathrm{PW}_{p, \lambda}(r, s)$, so that
\[
\liminf_{n, m \to \infty}
P(\mathrm{PW}_{p, \lambda}(r, s) \in C_{mn})
\ge
1 - \alpha.
\]
The previous distributional results allow us to construct $C_{nm}$. Although we focus on the regularized PRW distance, we can also construct such an interval for the IPRW distance under $r \neq s$ in the same manner. 

Let $\hat{r}_\ell^{\ast}$ and $\hat{s}_\ell^{\ast}$ be the empirical bootstrap distributions obtained from bootstrap samples $X_1^\ast, ..., X_\ell^\ast \sim \hat{r}_n$ and $Y_1^\ast, ..., Y_\ell^\ast \sim \hat{s}_m$, respectively. 
We denote the $\alpha/2$ and $(1 - \alpha/2)$ quantiles of $\mathrm{PW}_{p, \lambda}(\hat{r}_\ell^\ast, \hat{s}_\ell^\ast)$ as $q_{\alpha/2}$ and $q_{1-\alpha/2}$, respectively, and define
\[
C_{nm}
=
\left[
\mathrm{PW}_{p, \lambda}(\hat{r}_n, \hat{s}_m) - \sqrt{\frac{n+m}{nm}} q_{1- \alpha/2}, 
\mathrm{PW}_{p, \lambda}(\hat{r}_n, \hat{s}_m)- \sqrt{\frac{n+m}{nm}} q_{\alpha/2}
\right].
\]
Then, the validity of the rescaled bootstrap for the regularized PRW distance (Proposition \ref{bootprw}) implies that, under $\ell \to \infty, \ell/n \to 0, \ell / m \to 0$ as $n, m \to \infty$ and $m/(n+m) \to \delta \in (0, 1)$, $C_{nm}$ is an asymptotic $(1-\alpha)$ confidence interval for $\mathrm{PW}_{p, \lambda}(r, s)$.  

We apply the proposed interval estimation method to handwritten letter images from the Modified National Institute of Standards and Technology database (MNIST) dataset (\url{http://yann.lecun.com/exdb/mnist/}).
The dataset contains images with $576$ pixels for handwritten digits from $0$ to $9$. Because the distributions generating the images of each digit are likely to have low-dimensional structures, the PRW distance is expected to capture the differences between them effectively. Based on the above result, we construct $0.95$ confidence intervals for regularized PRW distances between pairs of digits. 
Specifically, we use $n=m=892$ images of digits 0, 1, 4, 7 and 9, and extract $128$-dimensional features of each image using a convolution neural network (CCN), as outlined in \cite{lin2020projection}.
Then, we estimate the global intrinsic dimension of feature data using the maxLikLocalDimEst function in the \textsf{R}~package \textit{intrinsicDimension} \cite{r_intrinsic} and obtain an estimate of 6.77. 
Based on this estimate, we set the projection dimension to $7$ and the order to $p=2$. 
We then construct the 0.95 confidence intervals using  $B=1000$ bootstrap with replacement $n^{4/5}  \approx 230$. 
The regularized PRW distance are calculated by the Riemannian optimization method proposed by \cite{lin2020projection}.

Figure \ref{intervals_MNIST} shows the results. The distances between digits $1$ and $7$ or digits $4$ and $9$ are smaller than those between digits $0$ and $1$ or digits $0$ and $4$.
Moreover, the distances between the same digits are quite small.
These results are consistent with our intuition. 

Furthermore, we add Gaussian noise with a standard deviation of $\sigma = 1, 5, 10$ to the feature data and again construct 0.95 confidence intervals for the regularized PRW distances. 
For comparison, we also construct 0.95 confidence intervals
for the original Wasserstein distances \citep{sommerfeld2018inference}.
The results are shown in Figure \ref{intervals_MNIST_noise}. 
The interval estimates of the regularized PRW distance are less influenced by the increase of the variance of the Gaussian noise than those of the Wasserstein distance. 
This result implies that, the PRW distance is more robust to the noise than than the original Wasserstein distance, when the dataset has a low-dimensional structure.

\begin{figure}[htbp]
   \centering
    \includegraphics[keepaspectratio, width=110mm]{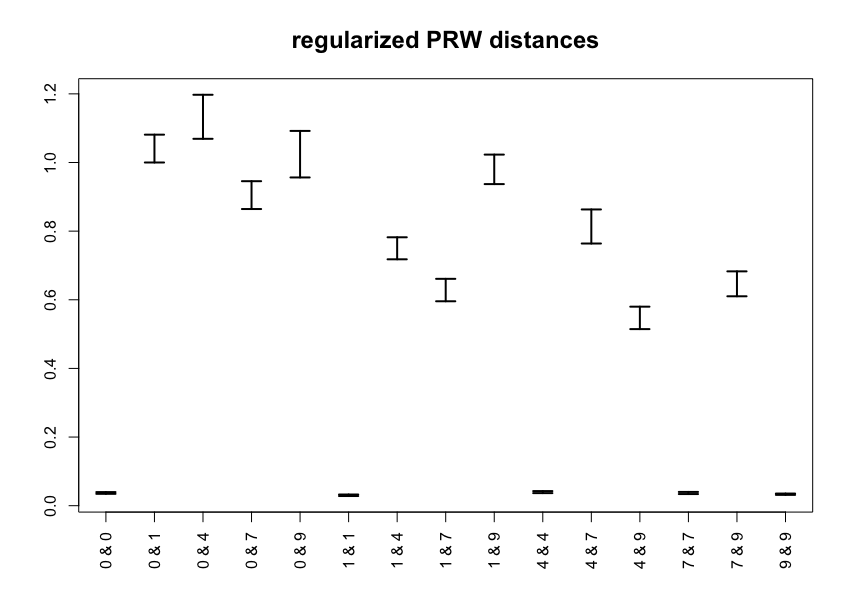}
    \caption{\small{Display of $0.95$ confidence intervals for the regularized PRW distance between hand-written digits. Intervals for the same digits are calculated by splitting the dataset into two groups. Intervals are normalized by setting the lower bound for $0 \& 1$ to be 1.}}
    \label{intervals_MNIST}
\end{figure}

\begin{figure}[htbp]
  \begin{minipage}[b]{0.9\linewidth}
    \centering
    \includegraphics[keepaspectratio, width=125mm]{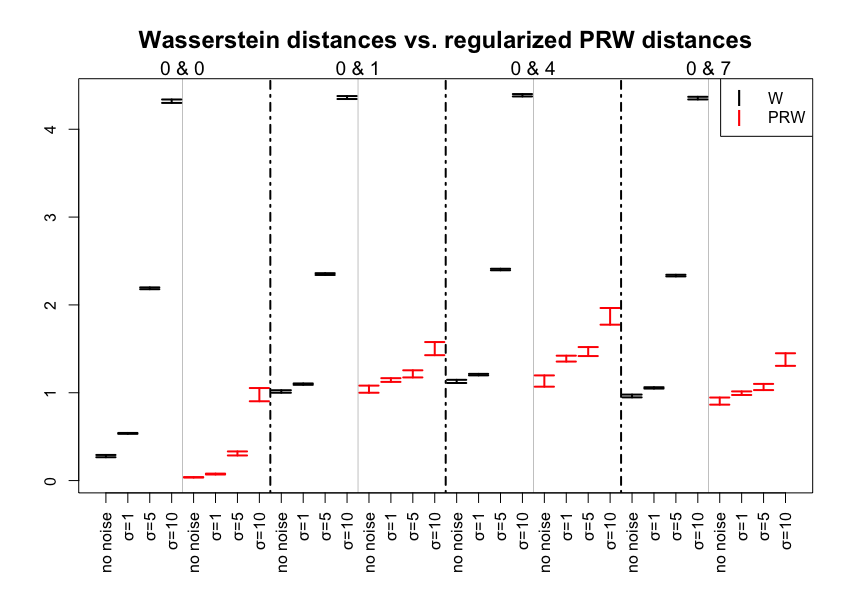}
  \end{minipage} \\
  \begin{minipage}[b]{0.9\linewidth}
    \centering
    \includegraphics[keepaspectratio, width=125mm]{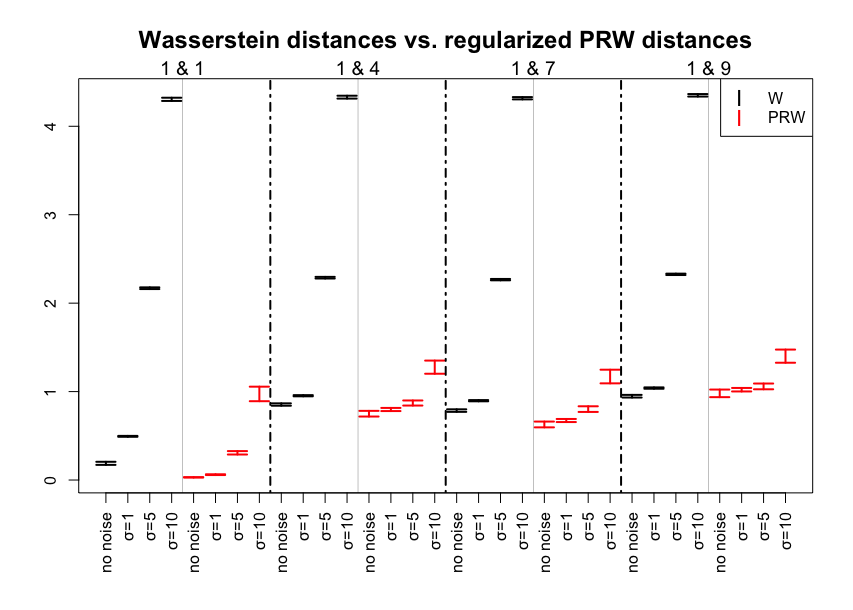}
  \end{minipage}
    \caption{\small{Display of $0.95$ confidence intervals for the regularized PRW and Wasserstein distance between hand-written digits with Gaussian noises. Intervals for the same digits are calculated by splitting the dataset into two groups.
    For each distance, 
    intervals are normalized by setting the lower bound for $0 \& 1$ to be 1.}}
    \label{intervals_MNIST_noise}
\end{figure}

\section{Discussion and conclusions}
This study investigated statistical inference for the IPRW and regularized PRW distances. Although these projection-based Wasserstein distances are practical for many machine learning tasks, their inferential tools have not been well established. We derived the limit distributions of the empirical versions of these distances on finite spaces by showing their directional Hadamard differentiability.
We also show that, while the naive bootstrap fails for these distances, the rescaled bootstrap is consistent. 

There are promising directions for future research. Our theoretical results are limited to finitely supported measures and it is worthwhile to extend them to more general settings. The appropriate choice of the replacement number of the rescaled bootstrap or projection dimension of the PRW distance is important in practice. Developing data-driven methods to choose their values is an interesting direction for further research.

\bibliographystyle{plain}
\bibliography{reference}

\end{document}